\documentclass[final]{colt2026}

\usepackage{csquotes}
\usepackage{tikz}
\usetikzlibrary {matrix, positioning,calc ,patterns}
\usepackage[dvipsnames]{xcolor}
\usepackage[noend]{algpseudocode}
\newcommand{\redit}[1]{{#1}}

\newcommand{\SAT}{\textit{SAT}}

\newcommand{\EXP}{\mathbb{E}}
\newcommand{\PROB}{\textnormal{Pr}}

\newcommand{\Pru}[2]{\ensuremath{\PROB_{#1}\left[#2\right]}}
\newcommand{\Ex}[1]{{\EXP \sbr{#1}}}
\newcommand{\Exu}[2]{\ensuremath{\mathbb{E}_{#1}\left[#2\right]}}

\newcommand{\norm}[1]{||\,#1\,||}

\newcommand{\rbr}[1]{\left(\,#1\,\right)}
\newcommand{\abs}[1]{\left|\,#1\,\right|}
\newcommand{\sbr}[1]{\left[\,#1\,\right]}
\newcommand{\cbr}[1]{\left\{#1\right\}}

\usepackage{color-edits}
\addauthor{kb}{cyan}   
\addauthor{jo}{purple}   
\addauthor{aa}{blue}
\addauthor{ig}{orange}

\newtheorem{claim}[theorem]{Claim}

\newcommand{\rr}{\ensuremath{r}}

\newcommand{\nn}{\ensuremath{n}}

\newcommand{\I}{\ensuremath{I}}

\newcommand{\as}{\ensuremath{\sigma}}

\newcommand{\range}{\ensuremath{\{-1, 1\}}}

\newcommand{\Q}{\ensuremath{Q}}

\newcommand{\D}{\ensuremath{D}}

\newcommand{\B}{\ensuremath{\mathcal{B}}}

\newcommand{\qry}{\ensuremath{q}}
\newcommand{\LL}{\ensuremath{L}}

\newcommand{\cc}{\ensuremath{c}}

\newcommand{\SDN}{\textnormal{SDN}}

\newcommand{\R}{\mathbb{R}}

\newcommand{\E}{\mathbb{E}}
\newcommand{\cI}{\mathcal{I}}
\newcommand{\cA}{\mathcal{A}}
\newcommand{\cS}{\mathcal{S}}
\newcommand{\cD}{\mathcal{D}}
\newcommand{\cDq}{\cD_{\cQ}}
\newcommand{\cB}{\mathcal{B}}
\newcommand{\cQ}{\mathcal{Q}}
\newcommand{\cU}{\mathcal{U}}
\newcommand{\cP}{\mathcal{P}}
\newcommand{\xor}{\oplus}

\newcommand{\bI}{\mathbf{I}}
\newcommand{\bC}{\mathbf{C}}
\newcommand{\bx}{\mathbf{x}}
\newcommand{\EqComment}[1]{\text{\emph{(#1)}}}

\newcommand{\code}{\mathcal{C}}

\newcommand{\ignore}[1]{}
\newcommand{\FF}{F}

\title[Quiet Planting for $k$-SAT, Multiple Solutions of Arbitrary Geometry]{Quiet Planting for $\mathbf{k}$-SAT, \\
Multiple Solutions of Arbitrary Geometry} 
\usepackage{times}

\coltauthor{
 \Name{Ali Ahmadi} \Email{ahmadia@umd.edu}\\
 \Name{Kiarash Banihashem} \Email{kiarash@umd.edu}\\
 \Name{Iman Gholami} \Email{igholami@umd.edu}\\
 \Name{Mohammad Taghi Hajiaghayi} \Email{hajiagha@umd.edu}\\
 \Name{Jan Olkowski} \Email{jan.olkowski@gmail.com}\\
 \addr University of Maryland, College Park, Maryland, USA
}

\begin{document}

\maketitle

\begin{abstract}
    Recent work on ``quiet planting'' in combinatorial optimization aims to generate instances with a hidden solution that is hard to recover, typically by making the planted distribution statistically indistinguishable from uniform for specific algorithms, such as statistical queries. A prominent example is planted $k$-SAT, where $O(n^{k/2})$ clauses can be planted while maintaining indistinguishability from uniform instances, evidenced by prior hardness results which also align with findings in SAT refutation. Despite extensive research and practical use in benchmarking SAT solvers, the challenge of quietly planting multiple solutions while preserving hardness has remained an open problem. 
    
    This work initiates the study of quiet planting with an arbitrary number of solutions, proposing the first method to construct quiet planting distributions for $k$-SAT formulas that accommodate more than one solution. We provide statistical query lower bounds for distinguishing these planted instances from uniform ones, and our method allows for planting solutions with arbitrary geometric relationships, including varying Hamming distances. A key innovation facilitating multiple solutions is the ability to incorporate arbitrary correlations between variable selection in clauses and their negation patterns, departing from prior approaches. We also investigate the worst-case complexity of SAT by showing the difficulty in distinguishing satisfiable instances with numerous solutions from unsatisfiable ones, addressing an open problem of Hsieh, Mohanty, and Xu (CCC'22). From a technical standpoint, we generalize the concept of $(r-1)$-wise uniformness in clause distributions, proving hardness holds if the marginal distribution over negation patterns is $(r-1)$-wise uniform, and reveal a connection to binary linear codes, demonstrating how a $[k, t, r]$ code can guide the planting of up to $2^t - 1$ solutions on $k$ variables with $(r-1)$-wise uniform negation distributions.
\end{abstract}

\begin{keywords}
Satisfiablity, Quiet Planting, Statistical Query
\end{keywords}

\section{Introduction}
\label{sec:intro}

The $k$-satisfiability problem ($k$-SAT) involves determining whether a solution 
exists for a boolean formula with $n$ logical variables. The formula consists of $m$ clauses,
each containing $k$ different variables or their negations. 
\redit{A solution satisfies the formula if, in each clause, at least one variable or its negation evaluates to true.} The NP-completeness of $k$-SAT, for $k \ge 3$, is crucial for understanding computational complexity and demonstrating the difficulty of problems in NP. It finds applications in fields such as automated reasoning and cryptography. 
\redit{The efficacy of $k$-SAT in encoding decision problems has led to its widespread use in reduction techniques}, particularly in addressing problems like Max-SAT and Max Clique, highlighting its adaptability in tackling intricate computational challenges both theoretically and in real-world problem-solving contexts.
Beyond the complexity theory applications, $k$-SAT can be used to model many real-world problems spanning areas such as AI planning, physical systems, and cryptography.

Understanding the hardness of k-SAT is particularly important because it is a central benchmark for computational difficulty, used in the evaluation of SAT solvers\footnote{The planting method described in this paper has been used in the DARPA program called \href{https://www.darpa.mil/news-events/2021-10-04}{QuICC} (Quantum-Inspired Classical Computing), which aims to apply insights from quantum algorithms to develop Quantum-Inspired (QI) solvers for complex DOD optimization problems, targeting a reduction in computational energy. The SAT instances generated using the method in this paper not only provide theoretical results but also serve as effective tools for testing SAT solvers (Quantum-Inspired). Although many methods exist for generating ``hard'' test cases for SAT solvers, those generated by the planting method in this paper are considered among the most challenging instances generated for this project.} and as a basis for hardness assumptions in cryptography. 
A single instance cannot reliably demonstrate hardness because an algorithm can incorporate the specifics of that particular instance into its implementation. Consequently, efforts have focused on characterizing classes of hard instances.
One extensively studied class is derived from the random satisfiability problem (random $k$-SAT). This class consists of instances where $m$ clauses are independently drawn from a distribution over all clauses of $k$ literals, without repetitions, across $n$ variables. 

For the class of random $k$-SAT instances, the seminal result of~\cite{DBLP:conf/stoc/DingSS15} (STOC'15) shows that for sufficiently large $k$ there exists a threshold $\tau_k \in \mathbb{R}$ such that if $m$ is a random variable drawn from $\textnormal{Pois}(n \alpha)$ and $m$ clauses are subsequently drawn independently from the uniform distribution over feasible clauses, then the instance has at least one solution whp\footnote{An event occurs with high probability (whp) if its probability tends to one in the limit $n \rightarrow \infty$.} if $\alpha < \tau_k$, but has no solutions whp if $\alpha > \tau_k$. 
Comprehending the difficulty of various classes of random SAT involves focusing on several key aspects. Examples include the refutation problem, which aims to determine whether we can prove that a random SAT instance cannot be solved, and the search/recovery problem, which concentrates on reconstructing a solution from a random SAT instance. Studying these problems is essential for assessing computational boundaries and demonstrating the level of difficulty inherent in various classes of instances.

In this context, the concept of the planted $k$-satisfiability problem (planted $k$-SAT) arises. This version of $k$-SAT follows the same basic setup as random $k$-SAT but with a crucial difference: every instance created is guaranteed to have at least one specific solution, meaning a fixed assignment is always a solution, though other valid solutions may also exist. 
\redit{Naively planted instances for $k$-SAT, in which a solution is fixed and the distribution is uniform over clauses satisfied by this solution, are nevertheless solvable in linear time. This can be done by a majority vote over the number of clauses that evaluate to true or false for each variable.}
\redit{However, some planted k-SAT distributions closely resemble uniform k-SAT distributions,  making them difficult to distinguish despite fundamental differences.} These distributions, often called \textit{quietly planted distributions}, provide solvable instances that mimic the behavior of uniform random $k$-SAT at certain thresholds—especially where uniform random $k$-SAT is typically unsatisfiable. 
While having universal results about the hardness of this class of instances remains an open problem, significant progress has been made in developing hardness in restricted models of computation based on plausible conjectures in complexity theory.

In this work, we focus on the statistical query model of computation introduced by ~\cite{DBLP:conf/stoc/Kearns93,DBLP:journals/jacm/Kearns98} (STOC'93). In this model, we are given access to a statistical oracle capable of as
estimating expectations of functions of random clauses without explicit
access to each clause. 
Focusing on these algorithms allows us to provide formal hardness results for the 
 the problem of distinguishing between a planted distribution from some reference distribution, typically taken to be the uniform distribution.
Formally, considering two probability distributions, denoted as $\Q$ and $\D$, over the set of clauses with $k$ variables, the planted distribution $\Q$ is subject to specific conditions, while the reference distribution $\D$ represents the baseline or null distribution.
The goal of the corresponding testing problem is to devise a statistical query algorithm capable of determining, based on queries to the statistical oracle, whether the underlying distribution is closer to $\Q$ (planted distribution) or $\D$ (reference distribution). The model of computation is very general, as many algorithms from various domains, including machine learning, cryptography, and anomaly detection, where distinguishing between different distributions is essential for making informed decisions based on observed data, can be implemented in this model. 

Since the introduction of the statistical query\redit{, significant results} have extended the understanding of the difficulty of planted $k$-SAT in this model. 
The distribution complexity parameter 
$\rr$, a key focus, characterizes the number of \redit{clauses} needed for efficient recovery of the planted assignment. 
Subsequent works have focused on developing algorithms depending on $\rr$ that efficiently identify the planted solution~\cite[]{DBLP:conf/approx/BogdanovQ09, DBLP:conf/focs/CharikarW04, DBLP:journals/cc/Applebaum16}. Significant advancement in this field was achieved through the work of \cite{DBLP:conf/stoc/FeldmanPV15, DBLP:journals/siamcomp/FeldmanPV18}. Their study focuses on the complexity of distinguishing between a planted $k$-SAT distribution and a random $k$-SAT distribution for the class of \emph{statistical query} algorithms.
The lower bounds in \cite{DBLP:conf/stoc/FeldmanPV15} effectively justify the barrier of $n^{k/2}$ clauses
observed for recovering planted solutions, as well as the closely related refutation problem.

The barrier of $n^{k/2}$ is a known limit for many problems arising around $k$-SAT satisfiability\redit{,} such as recovery and refutation\redit{. Many works} have obtained lower bounds showing that specific classes of algorithms cannot go beyond this barrier. In addition to the lower bound for statistical query implied by \cite{DBLP:conf/stoc/FeldmanPV15} (STOC'15), algorithms for the refutation problem were introduced by \cite{DBLP:conf/focs/AllenOW15,DBLP:conf/stoc/RaghavendraRS17} (FOCS'15, STOC'17), matching the lower bound established in \cite{DBLP:conf/stoc/KothariMOW17} (STOC'17).
In this work, we follow the model of computation presented in~\cite{DBLP:conf/stoc/FeldmanPV15}
while at the same time proposing novel and significantly more general results\redit{,} removing many limitations of the previous methods.
While our result relies on the SQ framework of~\cite{DBLP:conf/stoc/FeldmanPV15}, extending it to low-degree tests, directly or via the equivalences in~\cite{DBLP:conf/colt/BrennanBH0S21}, is a promising future direction. It is worth noting that they use a stronger notion of statistical dimension.

We note that the recovery problem is closely related to the refutation problem, in which the goal is to provide an algorithm that can refute the existence of a solution for a random instance (i.e., prove that the instance has no solution) with high probability. 
\cite{certsol} (CCC'22) \redit{studies} a version of this problem with multiple solutions, focusing on the problem of certifiable upper bounds for random SAT instances. However, this does not capture \redit{the aspect of arbitrary geometry, which we consider in this work. Moreover,} in general, lower bounds for refutation do not imply lower bounds for recovery (e.g., see \cite{DBLP:conf/stoc/Wein23,  DBLP:conf/soda/BanksMR21,
DBLP:conf/colt/BandeiraBKMW21,
DBLP:conf/innovations/BandeiraKW20}).

Furthermore, \cite{certsol} constructs a quiet planted distribution of $k$-SAT clauses based on a random \( k \)-uniform hypergraph with a planted independent set. In this distribution, the clauses are satisfied by \( 2^{|S|} \) assignments, where \( S \) is an independent subset of variables. Therefore, the solution space is determined by all possible values for variables of $S$. Although their approach is outlined without formal proof and it is not clear how large $|S|$ can be compared to $n$ while maintaining indistinguishability, a notable distinction is that their construction does not incorporate the geometry of the solutions. Another difference is that they focus on selecting variables to achieve multiple solutions, whereas in our paper, for a randomly chosen subset of variables, we plant multiple solutions with arbitrary geometry using the negation pattern. \redit{This matters in practice because overly structured or overly similar planted solutions can be exploited by heuristics, making benchmarks less informative. By generating hard instances with many planted solutions and controlled geometry, our approach improves both solver evaluation and our understanding of where hardness comes from.}
\redit{\paragraph{Our Contributions.} This paper presents two key results: First, we introduce a novel construction that quietly plants multiple solutions in \(k\)-SAT instances. Our method generalizes existing quiet planting techniques by allowing arbitrary geometry of solutions, leveraging binary linear codes and \((r-1)\)-wise uniformness. Second, building on the statistical query framework, we establish hardness guarantees for our construction by extending existing approaches. These advances significantly expand the applicability of quiet planting methods for generating hard \(k\)-SAT instances.

We begin by selecting a binary linear code. Assume this binary linear code is $[k, t, r]$, which has block size $k$, dimension $t$, and Hamming distance between codes of $r$ (see Section \ref{apx:prelim}). Next, a solution set $A$ is chosen such that its size is bounded by $2^t - 1$, ensuring the set remains small relative to the code’s dimension. Using the generator matrix, for any choice of $k$ variables $I$, we then define sign pattern distributions $Q_I$, which capture the distributional behavior of the code under different subsets. A central step is proving that these distributions exhibit $(r-1)$-uniformness, meaning that they retain uniformity properties up to dimension $r-1$. Finally, leveraging this structure, we establish a lower bound on the power of statistical query algorithms when applied to instances generated according to this distribution, highlighting inherent limitations in their ability to solve such problems efficiently.

Our result generalizes the work of~\cite{DBLP:conf/stoc/FeldmanPV15} by allowing the planting of multiple solutions with arbitrary geometry. More specifically, their work focused on the case \( r = k \), which corresponds to a planted random \( k \)-XOR-SAT instance, and proved a similar lower bound to the one we obtain in this paper. It is important to note that instances constructed by both our method and the prior approach can be solved using Gaussian elimination-based methods; however, introducing a small amount of noise can mitigate this vulnerability.}

\section{Preliminaries}
\label{sec:prelim}

\paragraph{General Notations.} We write $[n]=\{1,\dots,n\}$. For a vector $a$ of length $n$ and an index $i\in [n]$, we denote its $i^\text{th}$ coordinate by $a[i]$. For a sequence $a$ of length $n$ and any index sequence $b\in [n]^\ell$, we write
\(
a[b]=(a[b[1]],\dots,a[b[\ell]]).
\)
For binary vectors $u,v\in\{0,1\}^n$, the coordinate-wise ``XOR" is denoted by $u\oplus v$, and $-u$ denotes the coordinate-wise ``NOT''. Moreover, for any set $S\subseteq[n]$, we define
\(
\chi_S(u) := \bigoplus_{i\in S} u[i].
\)

In many parts of the analysis, we work with vectors in $\{-1,1\}^n$ (where the correspondence is $1\leftrightarrow 0$ and $-1\leftrightarrow 1$). In this setting, we extend the $\oplus$ operator so that for any $u,v\in\{-1,1\}^k$, 
\(
u\oplus v\quad\text{is their coordinate-wise product,}
\)
and $-u$ flips each sign.

\paragraph{SAT Notation.}  
A $k$-SAT formula is a CNF boolean formula in which each clause contains exactly $k$ literals; in a $k$-XOR-SAT, literals within a clause are connected via XOR. We let $n$ denote the total number of variables and $k$ the clause width. Each clause is represented as a tuple 
\(
C=(I,x),
\)
where,
    $I\in [n]^k$ is a sequence (or set) of $k$ distinct variables (the \emph{variable sequence}), and 
    $x\in\{0,1\}^k$ is the \emph{sign pattern}.
We denote by $X_k$ the set of all clauses of size $k$, and by $\cI_k\subseteq [n]^k$ the set of all variable sequences. Bold letters (e.g., $\bx$) indicate random objects, while non-bold letters denote fixed values.

An assignment is a vector in $\{0,1\}^n$. A clause $(I,x)$ is said to be \emph{satisfied} by an assignment $\sigma$ if there exists an $i\in [k]$ such that $x[i]=\sigma[I[i]]$. For an assignment $\tau\in\{0,1\}^n$, define
\(
C\oplus \tau := (I,\, x\oplus \tau[I]),
\)
and for any index subset $S\subseteq [k]$, define the derived clause
\(
C[S] := (I[S],\, x[S]).
\)

When clauses (or their components) are sampled from a distribution $\cQ$, we extend the above notation naturally. For example, for an assignment $\tau\in\{-1,1\}^n$, $\cQ_{\tau} = \cQ\oplus \tau$ denotes the distribution of $C\oplus \tau$ with $C\sim\cQ$, and for a fixed $I$, $Q_I$ denotes the conditional distribution of the sign pattern given that the variable sequence is $I$.

\paragraph{Statistical Algorithms.}
Under the class of \emph{statistical algorithms}, we consider algorithms that access data through statistical oracles. Given a distribution $D$ over a domain $X$, we use the following statistical oracles (see, e.g., \cite{DBLP:conf/stoc/FeldmanPV15}):
\begin{itemize}
    \item \textbf{1-MSTAT$(L)$:} For a query function $h: X\to\{1,\dots,L\}$, the oracle returns $h(x)$ for a random $x\sim D$.
    \item \textbf{VSTAT:} Given a parameter $t>0$ and a query function $h:X\to[0,1]$, it returns a value $v\in [p-\tau,\, p+\tau]$, where
    \[
    p=\E_{D}[h(x)]\quad \text{and}\quad \tau=\max\left\{\frac{1}{t},\,\sqrt{\frac{p(1-p)}{t}}\right\}.
    \]
    \item \textbf{MVSTAT$(t,L)$:} For $t>0$, a query function $h:X\to\{0,\dots,L-1\}$, and a collection $\mathcal{S}$ of subsets of $\{0,\dots,L-1\}$, the oracle returns a vector $v\in\R^L$ such that for every $Z\in\mathcal{S}$,
    \[
    \left|\sum_{\ell\in Z} v_\ell - p_Z\right|\le \max\left\{\frac{1}{t},\,\sqrt{\frac{p_Z(1-p_Z)}{t}}\right\},
    \]
    where $p_Z=\Pru{x\sim D}{h(x)\in Z}$. The cost of a query is $|\mathcal{S}|$.
\end{itemize}
The $1$-MSTAT oracle, in the above general version, was defined by~\cite{DBLP:conf/stoc/FeldmanPV15}. Related versions can also be found in~\cite{DBLP:journals/jacm/FeldmanGRVX17,DBLP:conf/colt/BrennanBH0S21}. This oracle corresponds to a single query to a distribution where the output complexity is limited to $L$ values. 

A statistical algorithm's complexity is measured by the number of queries multiplied by the cost per query. This model is broad enough to capture many standard approaches\redit{,} including EM, MCMC, method of moments, simulated annealing, and convex optimization.

\paragraph{Distinguishing Distributions.} We study the following \emph{distributional decision problem} $\mathcal{B}(\mathcal{D},D)$: Given a reference distribution $D$ over a domain $X$ and a set $\mathcal{D}$ of alternative (or planted) distributions, decide whether an unknown distribution $D'\in \{D\}\cup\mathcal{D}$ is equal to $D$ or belongs to $\mathcal{D}$, using $t>0$ samples from $D'$. In many applications (e.g., in $k$-SAT), one sets $X=X_k$, takes $\mathcal{D}$ to be a collection of distributions with planted solutions, and $D$ as the uniform distribution over $X_k$.

For any function $h: X\to \R$, with $\|h\|_D=\sqrt{\E_{D}[h^2]}$, the \emph{discrimination norm} relative to $D$ is
\[
\kappa_2(\mathcal{D},D) = \max_{h:\,\|h\|_D=1} \E_{D'\sim\mathcal{D}}\Bigl[\,\bigl|\E_{D'}[h]-\E_{D}[h]\bigr|\Bigr].
\]
This norm measures the average effectiveness of the query function $h$ in distinguishing distributions from $\mathcal{D}$ from the reference $D$. In particular, as shown in \cite{DBLP:conf/stoc/FeldmanPV15}, if $\kappa_2(\mathcal{D},D)=\kappa$, then a single query to $\operatorname{VSTAT}(1/(3\kappa^2))$ is insufficient to differentiate a typical $D'\in\mathcal{D}$ from $D$.

\paragraph{Statistical Dimension.} The \emph{statistical dimension} $\SDN(\mathcal{B}(\mathcal{D},D),\kappa)$ is defined as the largest integer $d$ for which there exists a finite set $\mathcal{D}_D\subseteq\mathcal{D}$ such that every subset $\mathcal{D}'\subseteq\mathcal{D}_D$ with $|\mathcal{D}'|\ge |\mathcal{D}_D|/d$ \redit{satisfies
\(
\kappa_2(\mathcal{D}',D)\le \kappa.
\)}
Intuitively, a high statistical dimension indicates that even large subsets of $\mathcal{D}_D$ are hard to distinguish from $D$.

This concept immediately leads to lower bounds on the query complexity of statistical algorithms. In particular, \cite[Theorem~3.4]{DBLP:conf/stoc/FeldmanPV15} shows that if $\SDN(\mathcal{B}(\mathcal{D},D),\kappa)=d$ and $L\ge 2$, then any randomized statistical algorithm that solves $\mathcal{B}(\mathcal{D},D)$ with probability at least $2/3$ requires:
\begin{itemize}
    \item $\Omega(d/L)$ queries to \textit{MVSTAT}$(L,1/(12\kappa^2L))$, and 
    \item $\Omega\bigl(\min\{d,1/\kappa^2\}/L\bigr)$ queries to \textit{1-MSTAT}$(L)$.
\end{itemize}

\section{\redit{Quiet planting of multiple solutions}}
\label{sec:our-results}

Our main result in this paper is a method for generating random $k$-SAT instances with quietly planted solutions that are \enquote{hard} to distinguish from uniformly
random instances, where the notion of hardness
is formalized by studying the problem in the statistical query framework and characterizing its statistical dimension.
Specifically, we study algorithms that only have
indirect access to the underlying distribution of the clauses
using a statistical oracle, and provide
a lower bound on the number of queries they need to distinguish
the underlying distribution from a uniformly random one. 

\redit{Our approach is novel in two key aspects. First, it enables the planting of multiple solutions. Second, it allows the planted solutions to exhibit arbitrary geometry in $\mathbb{F}_2^n$. 
More precisely, by \emph{arbitrary geometry}, we mean that for any set $S \subseteq \mathbb{F}_2^n$, our method constructs a random instance whose solution set contains a subset $S' \subseteq \mathbb{F}_2^n$ such that there exists a bijection $\phi: S \to S'$ satisfying
\(
\phi(x) \oplus \phi(y) = x \oplus y \text{ for all } x, y \in S.
\)
In other words, the coordinate-wise XOR structure among the original vectors in $S$ is preserved under the mapping to the solution set.}

It is noteworthy that the concept of arbitrary geometry is crucial for effectively testing SAT solvers. Quiet planting is primarily motivated by the need to evaluate the robustness of these solvers; thus, in scenarios with nearly identical sets of solutions, the solvers may not be rigorously tested. Additionally, certain heuristics could potentially compromise the integrity of the test cases.
We also remark that planting an exact set of solutions into a random instance is, to some extent, an infeasible \redit{task}, as it \redit{can be defeated by a statistical algorithm that explicitly checks for} these very specific solutions.
Please note that the trivial way of planting a solution by randomly taking satisfying clauses is vulnerable to statistical query algorithms because there is a bias in the number of times that a variable appears as true or false regarding its value in the chosen assignment for planting.

In this paper, we follow the statistical query model used by \cite{DBLP:conf/stoc/FeldmanPV15} for studying planted $k$-SAT instances.
Their work showed that
the primary parameter measuring the hardness of the testing problem for distinguishing planted instances from uniformly random ones is the \emph{statistical dimension} of the corresponding testing problem.
Here, the statistical dimension, which we will also denote with \emph{SDN}, serves as a proxy for the hardness of the problem, with higher dimensions corresponding to harder problems (see section~\ref{sec:prelim} for a formal definition)
This is demonstrated in \cite{DBLP:conf/stoc/FeldmanPV15}, where the authors prove that bounds on statistical dimension imply lower bounds on the query complexity of algorithms for the problem for various statistical query oracles.
Given these results, we focus on specifying distributions for which the statistical dimension of this testing problem is high. Formally, we prove the following result (see the appendix for the formal proof).

\begin{theorem}\label{thm:intro_main}
  Assume that there exists a binary linear code \([k, t, r]\) and a variable \(L < 2^t\).
  Let $\cA$ be an arbitrary set of assignments of size at most $L$.
  A set $\cD$ of distributions over clauses exists such that
  \begin{enumerate}
      \item 
      The statistical dimension of the testing
      problem $\cB(\cD, \cU)$ of distinguishing a
      $\cU$ (Uniform distribution of $k$-clauses) from a distribution $\cQ$ chosen uniformly
      at random from $\cD$ satisfies
      the bound for any $q \ge 1$
    \begin{align}
        \SDN\left(\B(\cD, \cU), \frac{\cc(\log \qry)^{\rr/2}}{\nn^{\rr/2}}\right)\ge \qry,
        \label{eq:sdn_bound_intro}
      \end{align}
     where $\cc$ is a constant and SDN refers to the statistical dimension of the problem (see Section~\ref{sec:prelim})
      
      \item 
      For any $\cQ \in \cD$, there exists 
      $\tau \in \cbr{0, 1}^n$ such that
      any clause $C \sim \cQ$ satisfies
      all of the assignments $\sigma \xor \tau$ for $\sigma \in \cA$ with probability $1$, where $\xor$ denotes the coordinate-wise XOR between two vectors.
  \end{enumerate}
\end{theorem}

Intuitively, the bound for SDN suggests that any randomly chosen distribution from \(\cD\) remains indistinguishable from the uniform distribution over \(k\)-clauses up to \(n^{r/2}\) clauses.
In the second point, we focus on the arbitrariness of the geometry of our solutions. Ideally, for any \(\cA\), we would want to generate instances where the solutions are exactly \(\cA\) and are hard to distinguish from random for all algorithms. However, if \(\cA\) is prespecified, this is not possible, as an algorithm could simply check for the solutions in \(\cA\). Therefore, we adopt a modified approach, where instead of \(\cA\), the solution is a random XOR-rotation of \(\cA\). Intuitively, the geometry of the solution set remains unchanged, as all solutions are transformed using the same rotation. \redit{We emphasize that $\tau$ is fixed once per distribution $\cQ$, not per clause: a single $\tau$ is chosen for $\cQ$, and all clauses of the instance are then sampled i.i.d.\ from this fixed $\cQ$, so the planted set $\{\sigma\xor\tau : \sigma\in\cA\}$ is the global set of satisfying assignments.}

As an immediate result and evidence of tightness, we can achieve the bound provided in \cite{DBLP:conf/stoc/FeldmanPV15} by using the \emph{repetition code} as a linear code. Since a repetition code of length $k$ is a $[k, 1, k]$ binary linear code, we can apply the above theorem to conclude that planting one assignment is quiet up to $n^{k/2}$ clauses. \redit{Given the structure of the repetition code, planting a solution in our method is equivalent to planting one in a \(k\)-XOR-SAT instance.}

Another result can be achieved by applying Gilbert-Varshamov bound \cite[]{gilbert1952comparison, varshamov1957estimate} which proved the existence of binary linear code of $[k, (1-H(\delta))k, \delta k]$ for any $\delta < 1/2$ ($H(.)$ denotes the binary entropy function). Consequently, planting less than $2^{(1-H(\delta))k}$ solutions is quiet up to $n^{\delta k / 2}$ clauses. This means we can plant an exponential number of solutions with arbitrary geometry. 

As an immediate corollary, \redit{using known reductions in the literature,} the above theorem implies statistical
query lower bounds for various statistical oracles. We refer to Section~\ref{sec:prelim} for a formal definition of these oracles and an overview of the reductions we use
We note that the bounds we obtain here
correspond to the hardness bounds for distinguishing
random $r$-XOR-SAT from uniform in the same frameworks.
As previously noted by \cite{DBLP:conf/stoc/FeldmanPV15} and \cite{DBLP:conf/stoc/KothariMOW17}, these bounds
suggest that any efficient algorithm for
detecting the existence of our planted solutions requires at least
$\Omega(n^{r/2})$ clauses.

Here, we elaborate on the most challenging part of our findings, which is the ability to plant many assignments of arbitrary geometry. 
In particular, our method allows for planted solutions of arbitrary Hamming distance\redit{,} where Hamming distance refers to the number of variables for which two solutions differ. 

We first recap the state-of-the-art approach for generating quiet planting instances of $k$-SAT
for a single solution.
Let us view each clause $C$ as a tuple
$(I, x)$ where $I$ is an ordered set of $k$ variables and $x \in \cbr{0, 1}^{k}$ is a \emph{sign pattern} that shows whether a variable or its negation \redit{appears} in the clause.
We will often use bold letters to emphasize that a vector or a set is random (e.g., $\bx$ vs $x$).
Prior works construct a distribution over the set of all feasible clauses, denoted $X_{k}$, by combining two independent distributions. 
The first distribution, denoted by $I_{k}$, uniformly chooses an ordered set $\bI$ of $k$ variables with no repetitions. The second one, 
$Q$ is defined over the set of sign patterns $\cbr{0, 1}^{k}$, with the important property that the all ones vector 
$(1, \dots, 1)$ has probability $0$ under $Q$.
If the goal is to plant a reference assignment $\sigma$, a random clause is generated according to the distribution $(\bI, \sigma[\bI] \xor \bx)$ where $\bI \sim I_{K}$, $\bx \sim Q$.
Since $(1, \dots, 1)$ is not sampled under $Q$, this ensures that the clause is always satisfied under $\sigma$.
While this approach allows us to construct distributions hard to distinguish from the uniform random distribution in the statistical query model (cf. Theorem~3.1 in~\cite{DBLP:conf/stoc/FeldmanPV15}), it inherently cannot be used to plant another solution of a large Hamming distance to $\sigma$, as we note below.
\begin{proposition}\label{lem:santhos-impossibility}
For any $Q$, if $\sigma$ is a planted solution in the above construction, there is no $\sigma' \in \cbr{0, 1}^n$ having Hamming distance $k$ to $\sigma$ that is satisfied by every clause that belongs to the distribution specified by the above construction.
\end{proposition}
\begin{proof}
[Proof of Proposition~\ref{lem:santhos-impossibility}]
Assume that such $\sigma'$ exists. Let $x \in \cbr{0, 1}^{k}$ be any element from the support of $Q$. Since Hamming distance of $\sigma$ and $\sigma'$ is $k$, thus there exists an ordered set $I$ of different literals such that $\sigma(I[i]) = \sigma'(I[i])$ if and only if $x[i] = 1$. Since $x$ was chosen from the support, the clause $(\bI, \sigma[I] \xor x)$ has a positive probability of appearing in the random formula. However, based on the specific choice of the variables $I$,
if we XOR the sign patterns in this clause with the corresponding assignment under $\sigma'$, we obtain
$
    \sigma'[I] \xor \sigma[I] \xor x = \left(1, \ldots, 1\right)
    .
$
\redit{This means that no variable in the clause is satisfied under $\sigma'$.}
contradicting the assumption made at the beginning of the proof. 
\end{proof}

In contrast, in our approach for quiet planting, we deviate from using one distribution $Q$ in conjunction with all possible ordered sets of variables $I$. Rather, we allow the distribution of sign patterns to depend on $I$. This allows for better alignment of the sign patterns to the set of variables of a clause, which in turn allows us to plant assignments of arbitrary geometry, overcoming the counterexample of Proposition~\ref{lem:santhos-impossibility}. It also allows for the quiet planting of more than one solution. 
But at the same time, it poses new challenges of characterizing properties that would yield lower bounds on the statistical dimension of such a distribution and\redit{,} in consequence, using known reductions, would imply the hardness of distinguishing this distribution from the uniform one in the statistical queries model.

\section{Universal statistical lower bound for \boldmath $(r-1)$-wise uniform distribution}
Firstly, we discuss the property of a distribution over $k$-SAT clauses, called \textit{$(r-1)$-wise uniformness}, that is responsible for formalizing hardness.
This property assures that \redit{when fewer than $r$ variables are considered,} the marginal distribution on these variables is the uniform distribution. We refer the reader to Appendix~\ref{sec:cons} containing Definition~\ref{def:d-wise_independent} for a formal definition of this property.
Although our definition is based on the same notion in \cite{DBLP:conf/stoc/FeldmanPV15} and \cite{DBLP:conf/stoc/KothariMOW17} (in the former the notion appears as \enquote{distribution complexity}),
our definition generalizes the previous works by allowing for dependencies between the variables of a clause and possible negations added to the variables.
This change is crucial to present later a distribution that every clause is satisfied by a fixed $>1$ number of solutions\redit{,} which, as noted above, can not be done without the relaxed independence assumption. 

As the main result of this part, we prove that $(r-1)$-wise uniform is a property that implies a large statistical dimension of a distribution. As a large statistical dimension of a distribution yields lower bounds on the query complexity of statistical algorithms, this property is also ultimately responsible for the hardness of distinguishing the distributions constructed in this paper from the uniform random one.
In conclusion, we prove the following hardness result (See Appendix ~\ref{sec:lower-bound} for the formal proof).
This theorem generalizes Theorem 3.5 in \cite{DBLP:conf/stoc/FeldmanPV15}.
\begin{theorem}\label{thm:hardness-main}
  For every $(r-1)$-wise uniform distribution $\cQ$ over $X_k$,
  there exists a constant $\cc > 0$, depending on $k$, such that, for any $\qry \ge 1$,
  \begin{align*}
    \SDN\left(\B(\cDq, \cU), \frac{\cc(\log \qry)^{\rr/2}}{\nn^{\rr/2}}\right)\ge \qry 
  \end{align*}
  where $\cU$ is a uniform distribution over $X_k$ and $\cD_{\cQ} := \cbr{\cQ_{\tau}: \tau \in \{-1, 1\}^n}$ denote
the set of all XOR-rotation of distribution $\cQ$.
\end{theorem}
Combining the above theorem with the reduction results from~\cite{DBLP:conf/stoc/FeldmanPV15} yields hardness for statistical queries using up to $n^{r/2}$ queries (see Section~\ref{sec:prelim}).

\section{Construction of \boldmath $(r-1)$-wise uniform distributions}
The second part of our approach involves constructing an $(r-1)$-wise uniform distribution $\cQ$ over the set of clauses that satisfies a given set of assignments, called $\cA$.
Let us fix an ordered set of $k$ variables $I$ in the clause $(\bI = I, \bx) \sim \cQ$ and let $Q_I$ denote the distribution of sign patterns $\bx \in \cbr{0, 1}^k$ conditioned on $I$.
The standard approach in prior work for ensuring a fixed assignment $\sigma \in \{0, 1\}$ satisfied $(\bI, \bx)$ is to let $Q_I$ be the uniformly random distribution over all these elements of $\{0, 1\}^{k}$ that agree with $\sigma[I]$ on an odd number of coordinates.
Intuitively, such constructed distribution 
is $(k-1)$-wise uniform because sampling from $Q_I$ is equivalent to sampling any $k-1$ coordinates uniformly at random in an arbitrary order and choosing the last one to ensure the number of agreements with $\sigma[I]$ is odd.
By the last property, any $\bx$ sampled from this distribution is satisfied by $\sigma$ because the number of coordinates agreeing with $\sigma$ cannot be $0$.

To generalize this idea into multiple solutions $\sigma_1, \dots, \sigma_{t}$, a natural approach is to select multiple subsets of $[k]$ as $S_1, \dots, S_t$ and require that for each $i\in [t]$, the sign pattern $\bx$ agrees with $\sigma_i$ on an odd number of coordinates of $S_i$. With such an $\bx$\redit{,} it is clear that all assignments are satisfied. However, even if we pick $\bx$ uniformly at random from all such sign patterns, it is not clear whether the $(r-1)$-wise uniformness property holds; indeed, the sign pattern $x[S_i]$ is definitely not uniform because it only takes half of all possible values.
Therefore, the distribution cannot be $(r-1)$-wise uniform for $r \ge \min |S_i|$.

To solve the above issue, 
a natural idea is to choose $|S_i| \ge r$; however, this alone is not sufficient. For example, suppose $S_1$ and $S_2$ both have size $\ge r$ but differ by only one coordinate: $s_1 \in S_1 \setminus S_2$ and $s_2 \in S_2 \setminus S_1$. Assume $\sigma_1$ and $\sigma_2$ agree on $S_1 \cap S_2$. Since each $S_i$ must have an odd number of agreements with $\sigma_i$, and $\sigma_1, \sigma_2$ agree on the intersection, it is impossible for $x$ to agree with $\sigma_1$ on $s_1$ but disagree with $\sigma_2$ on $s_2$. Specifically, agreement on $s_1$ implies an even number of agreements on $S_1 \cap S_2$, which in turn implies agreement on $s_2$. Thus, the distribution is not 2-wise uniform, as the set $S' = \cbr{s_1, s_2}$ does not take all possible values.

Intuitively, the issue in the above example is excessive dependence among the sets $S_i$. We show that avoiding dependencies of size $r$ when choosing the $S_i$ is essentially equivalent to designing a linear code. Specifically, suppose there exists a $[k, t, r]$ binary linear code with generator matrix $V \in \{0,1\}^{t \times k}$ (a code with $k$-bit codewords, message length $t$, and minimum Hamming distance $r$). We use this code to construct a distribution $Q_I$ that is $(r-1)$-wise uniform and has $t$ distinct solutions by taking each $S_i$ to be the set of $1$-coordinates in the $i^\text{th}$ row of $V$.

We further propose a novel idea to extend the number of solutions from $t$ to up to $L = 2^t-1$. 
Our main insight is that for each row of $V$, instead of satisfying only one $\sigma_i$, we can satisfy at least half of the remaining solutions. 
This allows us to divide the number of unsatisfied solutions in each row, ensuring that all solutions are satisfied since $L \le 2^{t}-1$.
Figure \ref{fig:linear-code} illustrates the use of linear code to achieve a $(r-1)$-wise uniform distribution, satisfying $2^t-1$ solutions.

\begin{figure}[ht]
\centering
\begin{tikzpicture}[scale=0.6]
\matrix [matrix of math nodes,left delimiter={[},right delimiter={]}] (m) {
v_{1,1} & v_{1,2} & \ldots & v_{1,k} \\
v_{2,1} & v_{2,2} & \ldots & v_{2,k} \\
\vdots & \vdots & \ddots & \vdots \\
v_{t,1} & v_{t,2} & \ldots & v_{t,k} \\
};
\foreach \i/\txt/\margin in {1/V_1^T/0.3, 2/V_2^T/0.3, 3/\vdots/0.7, 4/V_t^T/0.3}
  \node[anchor=east] at ($(m-\i-1.west) - (+\margin, 0)$) {$\txt$};
\matrix [right=1em of m,matrix of math nodes,left delimiter={[},right delimiter={]}] (m2) {
x_{1} \\
x_{2} \\
x_{3} \\
\vdots \\
x_{k} \\
};
\matrix [right=2em of m2,matrix of math nodes,left delimiter={[},right delimiter={]}] (m3) {
y_{1} \\
y_{2} \\
\vdots \\
y_{t} \\
};
\node [right=0.25em of m2, yshift=0cm] {$=$};
\foreach \i/\txt/\col in {1/S_1/MidnightBlue, 2/S_2/BrickRed, 3/\vdots/black, 4/S_t/OliveGreen}
  \node[\col,anchor=east] (\i) at ($(m3-\i-1.east) + (+2, 0)$) {$\txt$};
\foreach \i/\txt/\color/\dx/\dy in {1/S_1/MidnightBlue/1.5/0.3, 2/S_2/BrickRed/1.3/1.1, 4/S_t/OliveGreen/1.9/0.4} {
  \draw[->,color=\color] (\i) -- ($(\i) - (1.5, 0)$);
  \draw[->,color=\color] (\i) -- ($(\i) + (\dx, -\dy)$);
};
    \def\r{1 and 2};
    \def\t{100}
    \node (c) at (11.2, 0) {};
    \draw (c) circle (1 and 2);
    \node at ($(c) + (0, 2.3)$) {$\sigma_1, \ldots, \sigma_{2^t-1}$};
    \draw[pattern=north west lines, pattern color=MidnightBlue!\t] ($(c) + (180:\r)$) arc (180:0:\r) --  ($(c) + (180:\r)$);
    \draw[pattern=north east lines, pattern color=BrickRed!\t] ($(c) + (210:\r)$) arc (210:180:\r) -- ($(c) + (0:\r)$) -- ($(c) + (0:\r)$) arc (0:-30:\r) -- ($(c) + (210:\r)$);
    \draw[pattern=north west lines, pattern color=OliveGreen!\t] ($(c) + (250:\r)$) arc (250:290:\r) --  ($(c) + (250:\r)$);
    \node at ($(c) - (0, 1.25)$) {\vdots};
\end{tikzpicture}
\caption{This figure illustrates how a $[k, t, r]$ linear code results in a $(r-1)$-wise uniform distribution $Q_I$. The left matrix shows a code matrix $V$. For each row of $V$ denoted as $V_i^T$, the set of 1 coordinates of $V_i^T$ is $S_i$. For $S_i$, we can ensure that $x$ agrees with half of the remaining solutions by setting $y_i$ to a proper value. It is shown in the figure that the first at least half of solutions are satisfied by \textcolor{MidnightBlue}{$S_1$}, the next at least half of the remaining solutions are satisfied by \textcolor{BrickRed}{$S_2$}, and so on. Therefore, this approach supports up to $2^t-1$ solutions.}
    \label{fig:linear-code}
\end{figure} 

Combining these two, we \redit{present} the following theorem (with the proof in Appendix~\ref{sec:cons}):
\begin{theorem}\label{lem:marginal-to-codes-main}
    Suppose a $[k, t, r]$ linear code exists. Then for any set $I$ of $k$ variables and any collection $\cA$ of assignments with $|\cA| \le 2^t - 1$, there exists a distribution $\cQ_I$ over sign patterns on $I$ such that:
    \begin{enumerate}
        \item[(i)] For every sign pattern $x$ in the support of $\cQ_I$, the clause $C = (I, x)$ is satisfiable by some assignment in $\cA$, and
        \item[(ii)] If $x \sim \cQ_I$, then the marginal distribution on any subset of $r - 1$ coordinates of $x$ is uniform.
    \end{enumerate}
\end{theorem}

We can now apply any linear code to \redit{balance} the number of possible assignments and the complexity of the distribution. The Gilbert-Varshamov bound for linear codes leads to the following bound.

\begin{theorem}
For a given $k$, $\delta_0$, where $0 \le \delta_0 \le \frac{1}{2}$, and a given set of assignments $\mathcal{A}$ of size at most $2^{k(1 - H(\delta_0))} - 1$, there exists a  $(\delta_0 k - 1)$-wise uniform distribution $\mathcal{Q}$ over $k$-clauses, where $H(x) = - x \log_2(x) - (1-x) \log_2(1-x))$.
\end{theorem}
This result shows that a linear reduction in complexity allows an exponential expansion in the set of assignments. Although this theorem does not match the bound provided in \cite{DBLP:conf/stoc/FeldmanPV15}, our linear code planting technique allows planting a solution in up to $n^{k/2}$ clauses using the \emph{repetition code}, whose generator vector is $V = [1, 1, \dots, 1]$. The equation $Vx = y$ holds when the parity of $x$ matches $y$, equivalent to the XOR-SAT planting method in~\cite{DBLP:conf/stoc/FeldmanPV15}.

\section{Worst-case analysis of SAT instances with multiple solutions}

Lastly, we introduce a framework for assessing the hardness of SAT instances with numerous solutions, focusing on scenarios where the SAT problem is either unsolvable or has a predetermined minimum number of solutions. This is achieved by exponentially increasing the number of potential solutions and leveraging the Exponential Time Hypothesis (ETH), which is particularly effective in tackling the open problem posed by Hsieh, Mohanty, and Xu \cite{certsol} (CCC'22).

\begin{theorem}
\label{thm:worst-case-main}
Assuming the Exponential Time Hypothesis (ETH), there is no polynomial-time algorithm for solving \(\SAT(n, 2^{n - \log(n)^{1+c}})\) for any \(c > 0\), even when it is guaranteed that the input is (7/8 + \(\epsilon\))-satisfiable.
Note that $\SAT(n, s)$ is the class of SAT instances with $n$ variables that are either unsatisfiable or have at least $s$ solutions.
\end{theorem}
We further introduce a randomized algorithm for solving \(\SAT(n, 2^{n - \log(n)})\). The algorithm randomly selects assignments until a satisfying assignment is found or a predefined threshold of attempts is reached. If a solution exists, each random assignment satisfies the instance with probability \(\frac{1}{n}\).
By checking multiple random assignments, if we find a satisfying assignment, we conclude that the instance is satisfiable; otherwise, whp, it is unsatisfiable. Proofs are provided in Appendix~\ref{sec:worst-case}.

\section{Related work}
\label{sec:related}
Planting has recently been the subject of research in various problems, such as planting structures in graphs like cliques \cite[]{DBLP:conf/stoc/BuhaiKS23,DBLP:conf/colt/KothariVWX23,DBLP:conf/colt/BandeiraBKMW21} or planting solutions for satisfiability problems \cite[]{DBLP:conf/focs/GuruswamiHKM23}.
In the context of planted \( k \)-SAT, the development of hard-to-solve, quietly planted instances has motivated researchers from diverse fields, including theoretical computer science and statistical physics.
A common way to test SAT solvers is to use random instances—often generated by selecting an assignment $\sigma$ and randomly choosing clauses, rejecting those unsatisfied by $\sigma$. These instances are typically easy for local search approaches. Prior works~\cite{DBLP:conf/aaai/JiaMS05, barthel2002hiding}, experimentally studied generating quietly planted instances that are empirically challenging for such local methods or for other known $k$-SAT solvers. Another line of work has tried to build a theory for the hardness of instances with a single planted solution~\cite{DBLP:journals/corr/abs-0901-2130, DBLP:journals/jsat/KrzakalaMZ14, DBLP:conf/stoc/FeldmanPV15}. The work of~\cite{DBLP:conf/innovations/Blocki0DV17} extends this concept beyond predicates over $\mathbb{F}_2$.

Over the years, constructing planting distributions for the $k$-SAT problem with \textit{multiple} solutions has been explored experimentally~\cite{DBLP:journals/dke/LiuLY15, DBLP:journals/jair/ZhaoLLXJH23} and theoretically for up to two solutions or in hypothesis testing settings~\cite{DBLP:conf/aaai/AchlioptasJM04, DBLP:conf/aistats/BerthetE19}. However, our work is the first to plant more than two solutions of arbitrary geometry and establish theoretical hardness for this class of instances.

Other research efforts have focused on algorithms for the planted $k$-SAT problem. Spectral methods were studied in~\cite{DBLP:conf/soda/Flaxman03}. Bogdanov and Qiao~\cite{DBLP:conf/approx/BogdanovQ09} demonstrated the applicability of a Semidefinite Programming (SDP)-based algorithm by Charikar and Wirth~\cite{DBLP:conf/focs/CharikarW04}, which efficiently identifies the planted assignment for predicates lacking pairwise independence using only $O(n)$ evaluations. Their approach also extends to recovering inputs for $(r-1)$-wise (but not $r$-wise) independent predicates with $O(n^{r/2})$ evaluations~\cite{DBLP:journals/cc/Applebaum16}. These results enhance the toolkit for recovering planted assignments and deepen understanding of algorithmic strategies for planted $k$-SAT. The analysis also sheds light on challenges posed by planted instances, with implications for complexity theory, cryptography~\cite{DBLP:journals/iacr/Goldreich00a}, and algorithms for optimization problems with planted solutions~\cite{DBLP:conf/stoc/FeldmanGRVX13}. A key milestone by~\cite{DBLP:journals/corr/FeldmanPV14} unified prior efforts under a model covering specific planted satisfiability distributions. Their algorithm recovers planted solutions in both the stochastic block model and planted CSPs via a common generalization using random bipartite graphs. 

Finally, many works have developed the theory of hardness under other plausible conjectures in the complexity theory. The result of~\cite{DBLP:conf/stoc/KothariMOW17} gives a lower bound on the refutation of any CSP problem whose predicates are supported by $(r-1)$-wise uniform distributions for the class of sum of squares algorithms. This result is asymptotically tight due to the upper bounds given in~\cite{DBLP:conf/focs/AllenOW15, DBLP:conf/stoc/RaghavendraRS17}.

\acks{This work is partially supported by DARPA QuICC, DARPA expMath, ONR MURI 2024 award on Algorithms, Learning, and Game Theory, Army-Research Laboratory (ARL) grant W911NF2410052, NSF AF:Small grants 2218678, 2114269, 2347322.}

\bibliography{ref}

\appendix

\section{Prelimineries - Linear Codes}
\label{apx:prelim}
Linear codes, utilized in information and coding theory, are error-correcting codes crucial for dependable data transmission over noisy channels. Here, we present essential definitions concerning linear codes for future reference, an exhaustive summary can be found in~\cite{DBLP:books/lib/Berlekamp68}.

The formal definition of a code in information theory can be summarized as follows.
\begin{definition}[Code]
An error-correcting code $\code$ of length $n$ over a finite alphabet $\Sigma$ is a subset of $\Sigma^n$. The elements of $\code$ are termed codewords, and the
length $n$ of the codewords in $\code$ is the block length.

The alphabet of $\code$ is $\Sigma$, and if $|\Sigma| = q$, $\code$ is referred to as a $q$-ary code. For $q = 2$, it is a binary code. 

A code is associated with an encoding map $E$, which maps the message set $M$ of size $|\code|$ to codewords in $\Sigma^n$. The code is then the image of this encoding map.
\end{definition}
Here is the definition of the rate of a code:
\begin{definition}[Rate]
\label{def:lc-rate}
The rate of a code $\code \subseteq \Sigma^n$, denoted $R(\code)$, is defined by
$$R(\code) = \frac{\log |\code|}{n \log |\Sigma|}.$$
\end{definition}
The dimension of $\code$ is defined as $\log_{|\Sigma|} {|\code|}$. It is noteworthy that a $q$-ary code of dimension $l$ comprises $q^l$ codewords. In essence, $R(\code)$ represents the amount of non-redundant information per bit in the codewords of $\code$.

Next, we define the distance of a code as follows. We start by providing the definition of Hamming distance.
\begin{definition}[Hamming distance]
The Hamming distance between two strings $x$ and $y$ of equal length over a finite alphabet $\Sigma$, denoted $\Delta(x, y)$, counts the positions where the strings differ, i.e., $\Delta(x, y) = |\{i \,|\, x_i \neq y_i\}|$. The fractional Hamming distance or relative distance between $x$ and $y \in \Sigma^n$ is expressed as $\delta(x, y) = \frac{\Delta(x, y)}{n}$, ranging from 0 to 1 indicates the proportion of symbols that differ between the two strings, providing a normalized measure of their dissimilarity.
\end{definition}

\begin{definition}[Distance]
\label{def:lc-dis}
The distance of a code $\code$, denoted $\Delta(\code)$, is the minimum Hamming distance between two distinct codewords of $\code$, given by:

\[
\Delta(\code) = \min_{\substack{\code_1, \code_2 \in \code\\\code_1 \neq \code_2}} \Delta(\code_1, \code_2)
\]

In particular, for every pair of distinct codewords in $\code$, the Hamming distance between them is at least $\Delta(\code)$. The relative distance of $\code$, denoted $\delta(\code)$, is defined as the normalized quantity $\frac{\Delta(\code)}{n}$, where $n$ is the block length of $C$. Thus, any two codewords of $\code$ differ in at least a fraction $\delta(\code)$ of positions.
\end{definition}

Finally, we define a linear code.
\begin{definition}[Linear Code]
A linear code \(\code\) over a field \(\Sigma\) and of length \(n\) is a subspace of \(\Sigma^n\). For such a code, if the field size is \(q\), the code is termed a \(q\)-ary linear code. The notation \([n, k]_q\) represents a \(q\)-ary linear code of block length \(n\) and dimension \(k\), and if the code also specifies a minimum Hamming distance \(d\), it is denoted as \([n, k, d]_q\) code.
\end{definition}

In more detail, a field \(\Sigma\) is a set equipped with two operations (addition and multiplication) that adhere to the rules of arithmetic, including the presence of additive and multiplicative inverses for every element. A subspace of \(\Sigma^n\), where \(n\) represents the length of codewords, is a set of sequences over \(\Sigma\) that is closed under these operations, meaning any linear combination of elements (codewords) within the subspace remains in the subspace. The dimension \(k\) of a code \(C\) reflects the number of linearly independent vectors in \(C\), indicating the minimum set of vectors required to generate all elements of \(C\) through linear combinations.

For this paper, we only use binary codes, i.e.\redit{,} $\Sigma = \mathbb{F}_{2}$. Therefore, we omit the subscript in the notation of a linear code\redit{,} and we use $[n, k, d]:= [n, k, d]_{2}$.

\section{Construction of \boldmath$(r-1)$-wise uniform distributions - Proof of Theorem~\ref{lem:marginal-to-codes-main}}
\label{sec:cons}
In this part, we present a construction of a distribution $\cQ$ over the set of clauses $X_{k}$ that ensures the following two (informal) properties.
\newcommand{\numsol}{L}
\begin{enumerate}
  \item \textbf{Planted solutions.} Any instance $\FF$ consisting of clauses drawn from $\cQ$ has some $\LL$ distinct solutions $\sigma_{1}, \dots, \sigma_{\LL}$ with probability $1$. We will formalize the specific number of solutions later in this section. 
  \item \textbf{Indistinguishability.} It is \enquote{difficult} to distinguish $\FF$ from an instance generated uniformly at random provided that the number of clauses in $\FF$ is limited. We will formalize this notion by providing statistical query lower bounds for the corresponding testing problem in Section~\ref{sec:lower-bound}.
\end{enumerate}

For the first property
we require $\cQ$ to be supported
over clauses satisfied
by some fixed assignments $\sigma_1, \dots, \sigma_{\LL}$. For the second property to hold,
we require that for some parameter $r$, the marginal distribution of $\cQ$ on subsets of less than
$r$ variables are uniform.
Formally, we define the notion of $(r-1)$-wise uniform distributions as follows.
\begin{definition}\label{def:d-wise_independent}
We say a distribution $Q$ over $\range^k$ is \emph{$(r-1)$-wise uniform} if
for any subset $S \subseteq [k]$ of size $|S| < r$, the marginal distribution $Q[S]$ is uniform.
We say a distribution $\cQ$ over clauses is $(r-1)$-wise uniform if
the distribution of $\bI$ where $(\bI, \bx) \sim \cQ$ is uniform over $\cI$ and 
for any $I$, 
the distribution $Q_{I}$ is $(r-1)$-wise uniform,
where we recall that $Q_{I}$ denotes
the conditional distribution $(\bx \mid \bI=I)$ for
$(\bI, \bx) \sim \cQ$ (see Section~\ref{sec:prelim}).
The complexity of a distribution is the largest integer $r$ for which the distribution is $(r-1)$-wise uniform.
\end{definition}

Definition~\ref{def:d-wise_independent} generalizes
the definition of distribution complexity
in \cite{DBLP:conf/stoc/FeldmanPV15} and
$(r-1)$-wise uniformness in 
\cite{DBLP:conf/stoc/KothariMOW17}.
While these works only consider these definitions for distributions $Q$ over $\range^k$, we have extended it to distributions over clauses. 
For the special case studied by these works where $\cQ$ is a product distribution (i.e., $\bI$ and $\bx$ are sampled independently), the two definitions coincide. 

Utilizing the generality of the above definition, we will construct distributions according to the following recipe. First, a sequence of $k$ different variables of a clause, denoted by $\bI$\redit{,} is sampled uniformly, without replacement from the set $[n]$. 
Then, conditioned on the set $\bI$, a sign pattern $x$ is sampled according to a marginal distribution $Q_{I}$, to be specified later. 
The final clause is a tuple of the form $\bC := (\bI, \bx)$. The marginal distribution is chosen such that the final clause is satisfied by all the assignments from the predetermined set with probability $1$ and such that it is $(r - 1)$-wise uniform. The generality is crucially used in the fact that the distribution of the sign pattern can depend on the set of variables.

Formally, we can define the problem of finding the right distribution using linear programming. Fix an arbitrary set of assignments $\sigma_1, \dots, \sigma_{\LL} \in \{-1, 1\}^{n}$.
Definition~\ref{def:d-wise_independent} requires that for any $I$, the marginals of $Q_{I}$
over less than $r$ variables are uniform. 
This, in turn, is equivalent to the Fourier coefficients of $Q_{I}$
being zero for all non-empty subsets of size less than $r$, i.e.,
\begin{align}
  \hat{Q_I}(S) := \frac{1}{2^k}\sum_{x \in \{-1, 1\}^k} \rbr{Q_I(x) \chi_{S}(x)} = 0 \quad \forall  0 < |S| < r,
  \label{eq:fourier_zero}
\end{align}
where $\chi_{S}(x) := \xor_{s \in S} {x}$ denotes the XOR of $x[S]$.
In order to specify $\cQ$, it suffices to specify a distribution $Q_I$ for any $I$
such that:\\
\vspace{1mm}\noindent \textbf{(a)} with probability $1$ the clause $(I, \bx)$ for $\bx \sim Q_I$ satisfies
$\sigma_1, \dots, \sigma_{\LL}$,\\
\vspace{1mm}\noindent \textbf{(b)}
$Q_I$ satisfies Equation~\eqref{eq:fourier_zero}.\\
\vspace{1mm}\noindent Condition \textbf{(a)} is equivalent to requiring
the probability assigned under $Q_I$ to $-\sigma_i[I]$  to be $0$
as this is the unique sign pattern for which $\sigma_i$ is not feasible.
As for \textbf{(b)}, if we view $Q_I(x)$ as variables for $x \in \cbr{-1, 1}^{k}$, then
\eqref{eq:fourier_zero} turns into a linear equation for $Q_I(x)$. 
Putting everything together,
we formulate the problem of finding a suitable $Q_I$ as the following linear program where
the variables are $Q_I(x)$ for $x \in \{-1, 1\}^k$:
\begin{align}
\centering
\text{Find a distribution}\quad \quad &  Q_{I} \tag{LP1} \label{probl:LP}  \\
 \text{subject to}\quad \quad &  \underset{0 < |S| < r }{\text{\Large $\forall$}} \redit{\hat{Q}_I}(S) = 0, \label{lp:line-1}
 \\
 \notag
 &  \underset{i \in [\LL]}{\text{\Large $\forall$}} 
 Q_{I}(-\sigma_i[I]) = 0, \label{lp:line-2}\\
 &\sum_{x} Q_I(x) = 1 \text{ and } Q_I \ge 0   
 \notag
\end{align}

This way, in order to sample a clause $(\bI, \bx)$ from $\cQ$, we first sample $\bI$ uniformly at random.
We then solve linear program~\eqref{probl:LP} with $I$ set to $\bI$ and
sample $\bx$ from the resulting distribution.
We note that while the above linear program can always be formulated for any set of assignments $\sigma_i$,
it is not guaranteed to be \redit{feasible. For} instance, in the extreme case where
we consider all solutions in $\{-1, 1\}^n$, the program is always infeasible. 

In the next part of the section, we show that the problem is always feasible provided the number of solutions $\LL$
is \redit{at most $2^t - 1$}. The proof explicitly constructs a solution to the above linear problem by reducing it to the problem of designing a binary linear code. The remainder of this section is organized as follows. We first give an abstract framework of how\redit{,} based on a $[k, t, r]$ binary linear code\redit{,} one can construct a distribution $\cQ$ over $X_{k}$ that satisfies $2^t - 1$ arbitrary predetermined assignments. The distribution $\cQ$ will form a solution to the linear program~\eqref{probl:LP} assuming $\LL = 2^t - 1$. 
Finally, at the end of this section, we discuss Gilbert's construction of a binary linear code and how it influences bounds attainable by the construction.

\subsection{Linear code reduction}
\label{sec:linear-code-reduction}
In the light of the above reasoning, and specifically by the generality of the definition~\ref{def:d-wise_independent}, we can reduce the construction of the distribution $Q$ to constructing $(r-1)$-wise uniform distribution $\cQ_{I}$ assuming that $I$ is a given set of variables. In this part, we show how to solve the latter part by reducing it to the existence of a binary linear code with good properties. Specifically, let us fix a sequence of variables $I$. For a binary linear code with parameters $[k, t, r]$, our construction is such that $\cQ_{I}$ contains a planted set of $2^t - 1$ predetermined solutions $\mathcal{A} = \{\sigma_{1}, \ldots, \sigma_{2^t - 1}\}$ and is $(r - 1)$-wise uniform. We next formally state a lemma that describes the dependency between the existence of a linear code and properties of the distribution $\cQ_{I}$. The proof of this lemma is the central result of this section. A pseudocode incorporating both the reduction and the specifics of the construction conditioned on the sequence of variables $I$ is given in Figure~\ref{alg:clause-binary-codes}. 
To align better with notation used for linear codes, without loss of generality, in the section we treat the sign patterns $x$ in the clauses and the assignments $\sigma_1, \dots, \sigma_t$ as binary vectors over $\mathbb{F}_{2}$ field taking values in $\{0, 1\}^k$ and $\{0, 1\}^n$ respectively. 

\begin{lemma}\label{lem:marginal-to-codes}
    Let $I \subseteq [n]$ of size $k$ be a set of variables
    and let $\cA \subseteq \{0, 1\}^n$ be a set of assignments where $\abs{\cA} \le 2^{t} - 1$.
    Assume that there is a $[k, t, r]$ linear code. There exists a distribution $\cQ_I$ over $\{0, 1\}^{k}$ such that\\ 
    \vspace{1mm}\noindent\textbf{(i)} for any $x$ such that $\cQ_I(x) > 0$
    the clause $C = (I, x)$ is satisfiable under $\sigma_1, \dots, \sigma_{|\cA|}$, where $\{\sigma_1, \dots, \sigma_{|\cA|}\} = \cA$
    and\\
    \vspace{1mm}\noindent\textbf{(ii)} for any sequence $S$ of size $r - 1$ of elements of $[k]$, the marginal distribution of $x[S]$ for $x \sim \cQ_I$ is uniform.
\end{lemma}

We first give a sketch of the proof.
Let $V \in \{0, 1\}^{t \times k}$ denote the generator matrix of the binary linear code
and let $v_1, \dots, v_t$ denote the transposed rows of $V$, i.e. $v_{i} \in \mathbb{F}_{2}^{k}$. By definition of a $[k, t, r]$ linear code, any two vectors $u \ne v$ in the subspace of $v_1, \dots, v_t$ have Hamming distance $\ge r$. 

We set the distribution $\cQ_I$ to be uniform over the solutions of the linear system $Vx=b$ over $\mathbb{F}_2$, where the value of $b$ will be specified later. 
Each row in the linear system will correspond to a subset of the assignments, 
\redit{the} corresponding linear equation ensures feasibility for all of these assignments.
This effectively proves condition \textbf{(a)}. 
To prove condition \textbf{(b)}, we will rely on the properties of the linear code, namely the large Hamming distance between any linear combination of the rows, to prove that the solutions of the system of linear equations are uniform when focusing on any subset of size $<r$.

\begin{algorithm}[ht]
  \caption{Generating a clause of size $k$ that satisfies a set of given assignments}
  \label{alg:clause-binary-codes}
  \KwIn{Assignments $\mathcal{A}$, number of variables $n$, $t \times k$ matrix of a $[t, k, r]$ linear code}
  \KwOut{A random clause with $k$ variables that is satisfied by all assignments}

  \SetKwFunction{FMain}{GenerateClause}
  \SetKwProg{Fn}{Function}{:}{}
  \Fn{\FMain{$\mathcal{A}$, $n$, $V_{t \times k}$}}{
    Let $\I$ be a uniformly selected set of variables of size $k$\;
    
    Let $y = [0]^{k}$ be a vector to be filled in the construction\;
    
    $R_0 \gets \emptyset$\;
    
    \For{$i = 1$ \textbf{to} $t$}{
      Let $R_i^{(0)}$ and $R_i^{(1)}$ be empty sets of assignments\;
      
      \For{$\as \in \mathcal{A} \setminus R_{i-1}$}{
        \uIf{$\bigoplus_{j = 1}^{|I[v_i]| + 1} 1 \oplus \bigoplus_{j \in I[v_i]} a[j] = 0$}{
          Add $\as$ to $R_i^{(0)}$\;
        }
        \Else{
          Add $\as$ to $R_i^{(1)}$\;
        }
      }
      \uIf{$|R_i^{(0)}| > |R_i^{(1)}|$}{
        $R_i \gets R_{i-1} \cup R_i^{(0)}$\;
        
        $y[i] \gets 0$\;
      }
      \Else{
        $R_i \gets R_{i-1} \cup R_i^{(1)}$\;
        
        $y[i] \gets 1$\;
      }
    }
    Let $x$ be a uniformly random solution for the equation $A x = y$\;
    
    \Return{$C = (I, x)$ as a random clause}\;
  }
\end{algorithm}

We now proceed to a formal proof of Lemma~\ref{lem:marginal-to-codes}. (Described in Algorithm \ref{alg:clause-binary-codes}) For a binary vector $v$ of the same number of elements as the variables set $I$, we define $I[v]$ as the variable subset of $I$ consisting of elements on these positions $j$ that have $v[j] = 1$.
Define the following, inductive, partitioning of $\mathcal{A}$. Let $R_{0} = \emptyset$. For $i \ge 1$ and as long as $R_{i - 1} \neq \mathcal{A}$, we repeat the following operations. Partition the set of the assignments $\mathcal{A} \setminus R_{i - 1}$ into two subsets, those with the even number of positive literals on positions in $I[v_i]$ and those with the odd number of positive literals (Recall that for an assignment $a \in \{0, 1\}^{n}$, a literal on the position $i \in [n]$ is positive if and only if $a(i) = 1$). 
Let $\mathcal{A}_{i}$ be the subset with the larger number of assignments.
Define $R_{i} = R_{i - 1} \cup \mathcal{A}_{i}$ and let 
\begin{align*}
    y[i] := \bigoplus_{j = 1}^{|I[v_i]| + 1} 1 \oplus \bigoplus_{j \in I[v_i]} a[j] 
\end{align*}
for an arbitrary $a \in \mathcal{A}_i$.
Since all assignments have the same parity of positive variables on positions $I[v_i]$,  this definition does not depend on the choice of $a$\redit{,} and it proves its validity. Also, note that this construction halves the size of the set $R_{i}$ in each iteration. Since $\abs{\mathcal{A}} \le 2^t - 1$, thus after at most $t$ iterations the construction ends. Also, in the following, without loss of generality, we assume that the construction ends exactly after $t$ iterations; otherwise, we can use a subcode (i.e.\redit{,} a linear subspace) of the original code of the dimension equal to the desired number of iterations.
Thus, we get that the construction satisfies the following.
\begin{lemma}
Let $v_{1}, \ldots, v_{t}$ be a base of a $[k, t, r]$ binary linear code. The sequence of sets $\cA_{1}, \ldots, \cA_{t}$ and the sequence of numbers $y[i]$ give the following guarantees: \\
\noindent \textit{$(i)$} the family $\{\cA_{1}, \ldots, \cA_{t}\}$ is a partitioning of the set $\mathcal{A}$, \\
\noindent \textit{$(ii)$} for any $a \in \mathcal{A}$, let $i_{a}$ be the unique index such that $a \in \cA_{i_{a}}$. Then, it holds that
$$y[i_{a}] = \bigoplus_{j = 1}^{|I[v_{i_{a}}]| + 1} 1 \oplus \bigoplus_{j \in I[v_{i_{a}}]} a[j].$$
\end{lemma}
Let $y \in \mathbb{F}^{t}_{2}$ denote a vector having entries $y[i]$. 
Consider the linear system over $\mathbb{F}_{2}$ field 
\begin{align*}
v_{1}^{T} & x = y[1] \\
v_{2}^{T} & x = y[2] \\
& \vdots\\
v_{t}^{T} & x = y[t],
\end{align*}
which can be written $Vx = y$ in the matrix form. 
Finally, let us define the distribution $Q_{I}$ as the uniform distribution over the space of the solutions of the linear system $V \cdot x = y$.

\begin{lemma}\label{lem:nr_sols}
A clause $(\bI, \bx)$ where $x$ is the sign pattern drawn from $Q_{I}$ is satisfied by all assignments from the set $\mathcal{A}$ with probability $1$.
\end{lemma}
\begin{proof}
Let $x^* \in \mathbb{F}_{2}^{k}$ be any solution to the system $Vx = y$. Let $a \in \mathcal{A}$ be one of the assignments and denote $i^*$ be the index so that $a \in \mathcal{A}_{i^*}$. In particular, we get that $x^{*}$ satisfies
$$v_{i^*}^{T} x^* = y[i^{*}],$$
which means that
$$\bigoplus_{j \in I[v_{i^*} ]} x^{*}[j] =  \bigoplus_{1 = j}^{|I[v_{i^*}]| + 1} 1 \oplus \bigoplus_{j \in I[v_{i^*}]} a[j],$$
and as a consequence
\begin{equation}\label{line:1}
\bigoplus_{1 = j}^{|I[v_{i^*}]| + 1} 1 = \bigoplus_{j \in I[v_{i^*}]} \left( a[j] \oplus x^{*}[j] \right).
\end{equation}
This means that there must be at least one index $j \in I[v_{i^*}]$ such that the positivity/negativity of the variable on the position $j$ of $x^{*}$ matches the literal $a[j]$. Otherwise, 
$$\bigoplus_{j \in I[v_{i^*}]} \left( a(j) \oplus x^{*}(j) \right) = \bigoplus_{j \in I[v_{i^*} ]} 1 \neq  \bigoplus_{1 \le j}^{|I[v_{i^*}]| + 1} 1,$$
which contradicts \ref{line:1}. This proves that $a$ is satisfied by any clause belonging to the support of $Q_{I}$.
\end{proof}

The next lemma guarantees that the generated distribution is $(r-1)$-wise uniform.

\begin{lemma}\label{lem:Q_I_uniform}
The distribution $Q_{I}$ over sign patterns is $(r-1)$-wise uniform. 
\end{lemma}

We first provide an intuition for why linear codes show up.
Assume that the rows of $V$ do not satisfy the Hamming distance condition, i.e., there exists some $u\ne 0$ in the subspace of $v_1, \dots, v_t$ with Hamming distance $< r$ from the origin.
We will provide a simple proof that the distribution is not $(r-1)$-wise uniform. 

This means that there exists $\alpha\in \{0, 1\}^n$ such that
$u = \sum \alpha_i v_i$ .
Summing the rows of the linear system $Vx = y$ with coefficients $\alpha[i]$, we conclude that
for $x$ that is a solution to the system satisfies $u^{T} \cdot x = \sum \alpha[i] y[i]$. This means that the coordinates in $x$ corresponding to $u[i]$ always have a fixed sum\redit{,} and therefore are not distributed uniformly. Since $u$ has at most $(r-1)$ coordinates equal to 1, this means the distribution is not $(r-1)$-wise uniform.

The following lemma shows the converse of the above sketch, i.e., it shows that when the Hamming property does hold, the distribution is $(r-1)$-wise uniform. 

\begin{proof}[Proof of Lemma~\ref{lem:Q_I_uniform}]
To prove the theorem\redit{,} we calculate the number of solutions of the linear system $Vx = y$ for a vector $x$ that has any $r-1$ coordinates fixed. Denote the set of the $r-1$ positions $\ell_{1}, \ldots, \ell_{r-1} \in [k]$ and consider a fixed pattern of values $z_{1}, \ldots, z_{r-1} \in \{0, 1\}$ that must be preserved on these positions. To calculate the number of solutions we consider the following system of equations.
\begin{align*}
v_{1}^{T} & x = y[1] \\
v_{2}^{T} & x = y[2] \\
& \vdots\\
v_{t}^{T} &  x = y[t] \\
e_{\ell_{1}}^{T} & x = z[1] \\
& {\vdots} \\
e_{\ell_{r-1}}^{T} &  x = z[r-1],
\end{align*}
where $e_{\ell_{1}}, \ldots, e_{\ell_{r-1}}$ are the vectors from the standard basis. We denote the matrix form of the system as $V_{e} x = y_{e}$. 
We first observe the set of vectors $\{v_{1}, \ldots, v_{t}, e_{\ell_{1}}, \ldots, e_{\ell_{r-1}} \}$ is linearly independent.
\begin{lemma}\label{cla:lin-ind}
Vectors $v_{1}, \ldots, v_{t}, e_{\ell_{1}}, \ldots, e_{\ell_{r-1}}$ are linearly independent.
\end{lemma}
\begin{proof}
Consider any subset $B$ of the set of vectors and partition it into $B_{1} \subseteq \{ v_{1}, \ldots, v_{t} \}$ and $B_{2} \subseteq  \{e_{\ell_{1}}, \ldots, e_{\ell_{r-1}} \}$. It holds that $\{ v_{1}, \ldots, v_{t} \}$ is a linearly independent set of vectors, as it is the basis (codewords) of the linear code. Vectors $\{e_{\ell_{1}}, \ldots, e_{\ell_{r-1}} \}$ are linearly independent because they are a subset of the standard basis. 
Consider now a mixed linear combination $u_{1} + u_{2}$, where $u_{1} \in span\left( v_{1}, \ldots, v_{t} \right)$ and $u_{2} \in span\left( e_{\ell_{1}}, \ldots, e_{\ell_{r-1}} \right)$.
Since $u_{1}$ is a linear combination of codewords of the subspace $C$, it has at least $d$ distance from the codeword $(0, \ldots, 0)$, thus\redit{,} it has at least $d$ non-zero entries.
On the other hand, $dim\left(span\left( e_{\ell_{1}}, \ldots, e_{\ell_{r-1}} \right)\right) = r - 1$, which means that $u_{2}$ has at most $r-1$ non-zero entries.
Thus, it must be that $u_{1} + u_{2} \neq 0$ and the linear independence of the set of vectors $\{v_{1}, \ldots, v_{t}, e_{\ell_{1}}, \ldots, e_{\ell_{r-1}} \}$ follows.
\end{proof}

Using Lemma~\ref{cla:lin-ind}, we obtain that the matrix $V_{e}$ has full rank. Therefore, it must also be that the augmented matrix  $[V_{e}|y_{e}]$, derived from the linear system $V_{e}x = y_{e}$, is also a full rank matrix. We can apply then the Rouché–Capelli theorem to argue that the linear system $V_{e}x = y_{e}$ has exactly $|\mathbb{F}_{2}|^{k - t - (r- 1)} = 2^{k - t - (r-1)}$ solutions. Observe that the number is independent of the choice of the fixed positions and the pattern of $0$'s and $1$'s on these positions. In consequence, if $Q_{I}$ is the uniform distribution over the set of all solutions of $Vx = y$, it follows that
$$\Pru{x \sim Q_{I}}{x[L] = z} = \frac{2^{k - t - (r-1)}}{2^{k - t}} = \frac{1}{2^{r-1}},$$
which proves $(r-1)$-uniformness of $Q_{I}$.
\end{proof}

\subsection{Binary linear code construction}
The above construction relies on the relationship between the parameters $t$ and $r$ in the choice of the $[k, t, r]$ binary code. Using Gilbert-Varshamov's bound \cite{gilbert1952comparison, varshamov1957estimate}\redit{,} we can realize the trade-off such that the number of planted assignments is exponential in the distance of the binary code. 
Formally, Gilbert-Varshamov's theorem, based on a random construction, guarantees the following.

\begin{theorem}[Gilbert-Varshamov's theorem]
\label{thm:gv-lc}
Let $H(x)$ be the entropy function defined as $H(x) = - x \log_2(x) - (1-x) \log_2(1-x))$.
For any $k$, $0 \le \delta_0 \le \frac{1}{2}$, there exists a $[k, t, r]$ binary linear code $\code$ whose parameters $k, t, r$ satisfy $\delta(\code) = \delta_0$ and $R(\code) \ge 1 - H(\delta_0)$, where in this case $\delta(\code) = \frac{r}{k}$, $R(\code) = \frac{t}{k}$.
.
\end{theorem}

Using Lemma~\ref{lem:marginal-to-codes} together with the existence of this binary linear code gives the following result.

\begin{theorem}\label{thm:existence_cQ}
Fix $k$, for a given $\delta_0$, $0 \le \delta_0 \le \frac{1}{2}$, there exists a $(\delta_0 k - 1)$-wise uniform distribution $\cQ$ over clauses of size $k$ such that it satisfies a given set of assignments $\mathcal{A}$ of size at most $2^{k(1 - H(\delta_0))}-1$.
\end{theorem}
\begin{proof}
Consider a distribution $\cQ$ defined by first randomly and uniformly selecting a set of variables $\bI$ and then by applying Lemma~\ref{lem:marginal-to-codes} to the set of variables $\bI$ to generate a clause $(\bI, \bx)$. By Lemma \ref{lem:marginal-to-codes}, if the construction uses a code $\code$ with parameters $[k, t, r]$, and provided that $|\mathcal{A}| \le 2^{t} - 1$, then the distribution $\cQ$ is $(r-1)$-wise uniform and every assignment from $\cA$ is satisfied by the support of the distribution. 

To derive the achievable bounds on $k, t, r$\redit{,} we refer to the result of Theorem \ref{thm:gv-lc}. We first obtain the distance of the code. Theorem~\ref{thm:gv-lc} specifies that $\delta(\code) = \delta_0$. Employing Definition \ref{def:lc-dis} yields:
\begin{align*}
r &= \Delta(\code) \\
&= k\cdot \delta(C) \tag{Definition \ref{def:lc-dis}}\\
&= k\cdot\delta_0. \tag{Theorem \ref{thm:gv-lc}}
\end{align*}
As for the dimension of the code, we derive a bound on it from the bound on $R(\code)$. Theorem \ref{thm:gv-lc} asserts $R(\code) \ge 1 - H(\delta_0)$. By Definition~\ref{def:lc-rate}, we get that
\begin{align*}
t &= k\cdot R(\code) \tag{Definition \ref{def:lc-rate}}\\
&\ge k(1 - H(\delta_0)). \tag{Theorem \ref{thm:gv-lc}}
\end{align*}
Combining both bound together, we get that $\code$ can be a $[k, {k(1 - H(\delta_0))}, \delta_0 k]$ code. This code can cover at most $2^{k(1 - H(\delta_0))}$ solutions. \redit{Thus,} the theorem follows.
\end{proof}

Remark that for a fixed $\delta_0$ the number of solutions can be $2^{k(1 - H(\delta_0))} = 2^{\Theta(k)}$

\section{Universal statistical lower bound for \boldmath$(r-1)$-wise uniform distribution - Proof of Theorem~\ref{thm:intro_main} and~\ref{thm:hardness-main}}\label{sec:lower-bound}
In this section, we show how
to generate hard instances via access to an $(r-1)$-wise uniform distribution $\cQ$ over clauses of size $k$.
Our main result is the following theorem.
Specifically, we will show that if we randomly \enquote{rotate} the distribution $\cQ$ to obtain a new distribution $\cQ'$ and repeatedly sample from $\cQ'$, then the resulting instance will be 
hard. 
Our main result, Theorem~\ref{thm:hardness} below, bounds the statistical dimension of the testing problem of distinguishing $\cQ'$ from the uniform distribution $\cU$.
This further implies statistical query lower bounds for the problem (See Section~\ref{sec:prelim}).

Choose a \emph{rotation} vector $\tau$
uniformly at random from $\{-1, 1\}^{n}$
and
let $\cQ_{\tau} := \cQ \xor \tau$ denote
the distribution obtained by applying the XOR with $\tau$ on the clause obtained from $\cQ$ (see Section~\ref{sec:prelim} for the definition). 
Formally,
$\cQ_{\tau}$ is the distribution of the clause
$(\bI, \bx \xor \tau[\bI])$ where $(\bI, \bx) \sim \cQ$.
We generate the $k$-SAT instance $\FF$ by sampling $m$ clauses i.i.d from
$\cQ_{\tau}$.
A formal pseudo code is shown in Algorithm~\ref{alg:generate}.
Our approach is a generalization of
that of \cite{DBLP:conf/stoc/FeldmanPV15} and \cite{DBLP:conf/stoc/KothariMOW17}
who obtain instances by
fixing a distribution $Q$ over $\{-1, 1\}^{k}$ and
sampling clauses of the form
$(\bI, \bx \xor \tau[\bI])$ where
$\bI$ is a uniformly random subset of size $k$ from $[n]$
and $\bx$ is sampled from $Q$. 
Our model captures this special case by considering
a product distribution for $\cQ$; i.e., assuming that $\bI$ and $\bx$ are sampled independently.
Our approach is more general\redit{,} however as it allows
the sampled value of $\bx$ depend on the sampled value of $\bI$, which is crucial for generating
instances with multiple solutions\redit{,} as mentioned in the introduction.
\begin{algorithm}[ht]
  \caption{Generating an instance with $m$ clauses}
  \label{alg:generate}
  \KwIn{Distribution $\cQ$, number of clauses $m$, number of variables $n$, number of variables in each clause $k$}
  \KwOut{A SAT instance with $m$ clauses}
  
  \SetKwFunction{FMain}{GenerateSATInstance}
  \SetKwProg{Fn}{Function}{:}{}
  \Fn{\FMain{$\cQ$, $m$, $n$, $k$}}{
    Let $\tau$ be a uniformly random vector in the space $\{-1, 1\}^n$\;
    
    Initialize an empty SAT instance $\FF$\;
    
    \For{$1$ \textbf{to} $m$}{
      Sample a clause $(I, x)$ from $\cQ$\;
      
      \For{$i = 1$ \textbf{to} $k$}{
        $x[i] \gets x[i] \oplus \tau[I[i]]$\;
      }
      Add clause $C$ to $\FF$\;
    }
    \Return{$\FF$}\;
  }
\end{algorithm}

We next study the hardness of distinguishing the planted instance described
above from an instance drawn uniformly at random
by studying the problem in the statistical query model.
Specifically, for any algorithm that
has indirect access to the instance via a statistical query oracle for some distribution $\cP$,
we provide a lower bound on the number of queries it needs to
determine whether $\cP = \cU$ or $\cP \in \cQ_{\tau}$.
Formally,
let $\cD_{\cQ} := \cbr{\cQ_{\tau}: \tau \in \{-1, 1\}^n}$ denote
the set of all distributions $\cQ$.
Our main result is the following Theorem,
which can be seen as a generalization of Theorem 3.5 in \cite{DBLP:conf/stoc/FeldmanPV15}
who only consider the case where $\cQ$ is a product distribution.

\begin{theorem}\label{thm:hardness}
  For every $(r-1)$-wise uniform distribution $\cQ$ over $X_k$,
  there exists a constant $\cc > 0$, depending on $k$, such that, for any $\qry \ge 1$,
  \begin{align}
    \SDN\left(\B(\cDq, \cU), \frac{\cc(\log \qry)^{\rr/2}}{\nn^{\rr/2}}\right)\ge \qry 
    .
    \label{eq:sdn_bound}
  \end{align}
\end{theorem}

Combined with Theorem~\ref{thm:existence_cQ}, we obtain the following result.

\begin{theorem}\label{thm:main}
  Assume that $n, k, r, L$ satisfy
  $r\le k/2$ and $L \le 2^{k(1-H(\frac{r}{k}))}$, where $H$ is the binary entropy function.
  Let $\cA$ be an arbitrary set of assignments of size at most $L$.
  There exists a distribution $\cQ$ over $X_k$ such that
  \begin{enumerate}
      \item the testing problem
      $\cB(\cDq, \cU)$ where $\cDq = \cbr{\cQ \xor \tau: \tau \in \range^n}$ satisfies the bound in Equation~\eqref{eq:sdn_bound}.
      \item Any clause $C \sim \cQ$ satisfies
      all of the solutions in $\cA$ with probability $1$. 
  \end{enumerate}
\end{theorem}

Combining the above theorems with the reduction results obtained in~\cite{DBLP:conf/stoc/FeldmanPV15} 
(See Section~\ref{sec:prelim}) we get
\begin{theorem}\label{thm:queries_lower_bound}
Any randomized statistical algorithm that correctly decides the testing problem $\cB(\cDq, \cU)$ with a probability of at least $2/3$ (considering the random choice of $\cDq$ and the randomness of the algorithm) requires either:
\begin{itemize}
\item $m$ calls to the $\textit{1-MSTAT}(L)$ oracle, where $m \cdot L \ge c_1  \left(\frac{m}{\log n}\right)^r$ for a constant $c_1 = \Omega_k(1)$, or

\item $q$ queries to $\textit{MVSTAT}(L, \frac{c_2}{L},\frac{n^r}{(\log q) ^ r})$ for a constant $c_2 = \Omega_k(1)$ and any $q \ge L$.
\end{itemize}
\end{theorem}

\subsection{Proof of Theorem~\ref{thm:hardness}}
The proof of the theorem is obtained by simplifying and generalizing the proof of Theorem 3.5 in \cite{DBLP:conf/stoc/FeldmanPV15}. Their proof effectively decomposes the problem into special cases of $\ell$-XOR-SAT for $\ell \le k$ by considering
the contribution of each subset $S \subseteq [k]$ in the Fourier expansion of the above difference.
Specifically, they show that
$$
    \Exu{\cQ_{\tau}}{h} - \Exu{\cU{h}} =
    -2^{k}\sum_{S \subseteq [k]} \hat{Q}(S)
    \rbr{ \Exu{Z_{\ell, \tau}}{h_S} - 
    \Exu{\cU_\ell}{h_S}}
$$
where $Z_{\ell, \tau}$ is the $\ell$-XOR-SAT distribution for $\tau$ and $h_S$ is defined as a projection of $h$ into $X_\ell$. 
Such a decomposition is not possible in our case\redit{,} however\redit{,} because
the function $Q$ is not fixed and depends on $I$. As such, we cannot even define $\hat{Q}(S)$. 
To fix this issue, we opt for a more straightforward approach that directly handles the terms for each $S \subseteq [k]$. 

\begin{proof}[Proof of Theorem~\ref{thm:hardness}]
  \newcommand{\sumI}{\sum_{I \in \cI}}
  \newcommand{\sumx}{\sum_{x \in \range^{k}}}
  By definition of statistical dimension, it suffices to show that for any subset
  $\cD' \subseteq \cD_{Q}$ with size
  $|\cD'| \ge |\cD_{Q}| / q$ we have
  \begin{align*}
      \max_{h: \norm{h}_{\cU}=1}
      \Exu{\cQ' \sim \cD'}{
        \abs{\Exu{\cQ'}{h} - \Exu{\cU}{h}}
      }
      \le c \frac{(\log q)^{r/2}}{n^{r/2}}
      .
  \end{align*}

  We first give a sketch of the proof.
  Set $\tau$ to satisfy $\cQ' = \cQ_{\tau}$.
  The expectation
  $\Exu{\cQ_{\tau}}{h}$ can be written
  as $\sum_{I, x} \cQ_{\tau}(I, x) h(I, x)$.
  Let $Q_I$ denote the distribution, over sign patterns, of $\cQ$ after conditioning on the variable sequence.
  Formally,
  $Q_I(x) = \Pru{(\bI, \bx) \sim \cQ}{\bx = x \mid \bI=I}$.
  Since $\cQ$ was assumed to be $(r-1)$-wise uniform, for any fixed $I$ we know that
  $\cQ_{\tau}(I, x)$ is proportional to $Q_I(x \xor \tau[I])$;
  specifically, 
  $\cQ_{\tau}(I, x) = \frac{1}{\abs{\cI}} Q_I(x \xor \tau[I])$.
  We apply the Fourier
  transform
  on $Q_I$ to rewrite
  it as a (weighted) sum of 
  $\chi_{S}(x \xor \tau[I])$ terms for $S \subseteq [k]$.
  The term corresponding to $S \ne \emptyset$ cancels
  out with $\Exu{\cU}{h}$  because
  the Fourier coefficient for $\emptyset$ is $1/2^k$ for both $\cQ$ and $\cU$
  since they are both distributions.
  We can therefore
  rewrite the difference $\abs{\Exu{\cQ'}{h} - \Exu{\cU}{h}}$ as 
  $\abs{\sum_{S \subseteq [k], S \ne \emptyset} M_S(\tau)}$,
  where $M_S(\tau)$ denotes the sum of the terms corresponding to $\chi_S(x \xor \tau[I])$ in the expansion of $Q_I$.
  Crucially, we only need to consider $S$ of size at most $r$
  since the smaller $S$ do not appear because of the $(r-1)$-wise independence condition.

  Using Jensen's inequality, we show that it suffices
  to separately bound $\Exu{\cQ_{\tau} \sim \cD'}{\abs{M_S(\tau)}}$
  for each $S$ and then combine these bounds for a valid upper bound on
  $\Exu{\cQ_{\tau} \sim \cD'}{\sum_{S} \abs{M_S(\tau)}}$.
  This is formally shown in Equation~\eqref{eq:sdn_bound_decompose}.
  We observe that each $M_S(\tau)$ is a polynomial in $\tau$ because it is a sum of $\chi_S(x \xor \tau[I])$
  terms for different $i$, 
  and each such term can be written as
  $\chi_S(x) \prod_{i \in I[S]} \tau[i]$. 
  Using the fact that $\tau$ is chosen from a large
  set, specifically 
  $\cQ_\tau \sim \cD'$, we relate the expectation $\Exu{\cQ_\tau \sim \cD'}{\abs{M_S(\tau)}}$
  to the norm of the coefficients of $M_S$. 
  This norm can be bounded fairly easily\redit{,} however using the assumption 
  $\norm{h}_{\cU} = 1$.

  We proceed with a formal proof.
  We drop the subscript $k$ when it is clear from the context.
  Expanding the definition of $\cQ_{\tau}$ and applying
  the Fourier transform, we obtain
  \begin{align*}
    \Exu{\cQ_{\tau}}{h} 
    &=
    \sumI \sumx
    h(I, x) \cQ_\tau(I, x)
    \\&=
    \sumI \sumx
    h(I, x) \cQ(I, x \xor \tau[I])
    &\EqComment{Definition of $\cQ_{\tau}$}
    \\&=
    \frac{1}{\abs{\cI}}
    \sumI \sumx
h(I, x) Q_I(x \xor \tau[I])
    &\EqComment{Definition of $Q_I$}
    \\&=
    \frac{1}{\abs{\cI}}
    \sum_{I} \sum_{x}
    \sum_{S \subseteq [k]}
    h(I, x) \hat{Q}_I(S)  \chi_S(x \xor \tau[I])
    &\EqComment{Fourier expansion}
    .
  \end{align*}
  Let $\chi(x) = \xor_{i=1}^{|x|}x[i]$ for an arbitrary vector $x$ denote the XOR of $x$.
  Note that
  $\chi$ is multiplicative, i.e.,
  $\chi(x \xor y) = \chi(x)\chi(y)$ for any two vectors $x, y$ of same length.
  Grouping the above sum in terms of the variable $S$, we can rewrite it as 
  \begin{math}
    \sum_{S} M_S(\tau)
  \end{math}
  where
  \begin{align}
    M_S(\tau) &:= 
    \frac{1}{\abs{\cI}}
    \sum_{I, x}
    h(I, x)
    \hat{Q}_I(S) 
    \chi_{S}(x \xor \tau[I])
    \notag
    \\&=
    \frac{1}{\abs{\cI}}
    \sum_{I, x}
    h(I, x)
    \hat{Q}_I(S) 
    \chi((x \xor \tau[I])[S])
    &\EqComment{Definition of $\chi_S$}
    \notag
    \\&=
    \frac{1}{\abs{\cI}}
    \sum_{I, x}
    h(I, x)
    \hat{Q}_I(S) 
    \chi(x[S]) \chi(\tau[I[S]])
    &\EqComment{Multiplicativity of $\chi$}
    .
    \label{eq:M_S_sum}
  \end{align}
  Similarly, for $\cU$,
  \begin{align*}
    \Exu{\cU}{h} 
    &= 
    \sumI \sumx
    h(I, x) \cU(I, x)
    \\&=
    \frac{1}{2^k \abs{\cI_k}}
    \sumI \sumx
    h(I, x)
    &\EqComment{Definition of $\cU$}
    ,
  \end{align*}
  which is the same $M_{\emptyset}(\tau)$ because $\hat{Q}_I(\emptyset) = \frac{1}{2^k} \sum_{x}Q_I(x) = \frac{1}{2^k}$
  and $\chi(x[\emptyset]) = \chi(\tau[I[\emptyset]]) = \chi(\emptyset) = 1$. 
  It follows that
  \begin{align}
    \Exu{\cQ_{\tau} \sim \cD_{Q}}{
      \abs{\Exu{\cQ_{\tau}}{h} - \Exu{\cU}{h}}
    }
    \notag
    &=
    \Exu{\cQ_{\tau} \sim \cD_{Q}}{
      \abs{\sum_{S \ne \emptyset} M_{S}(\tau)}
    }
    \notag
    \\&\le
    \Exu{\cQ_{\tau} \sim \cD_{Q}}{
      \sum_{S \ne \emptyset} \abs{M_{S}(\tau)}
    }
    &\EqComment{Jensen's inequality}
    \notag
    \\&=
    \sum_{S \ne \emptyset} 
    \Exu{\cQ_{\tau} \sim \cD_{Q}}{
      \abs{M_{S}(\tau)}
    }
    &\EqComment{Linearity of expectation}
    \label{eq:sdn_bound_decompose}
    .
  \end{align}
  It\redit{,} therefore suffices to bound
  \begin{math}
    \Exu{\cQ_{\tau} \sim \cD_{Q}}{
      \abs{M_{S}(\tau)}
    }
  \end{math}
  for $S \ne \emptyset$. 
  To do this, we view $M_S(\tau)$ as a polynomial of degree
  $\ell := \abs{S}$ in $\tau$
  and bound the norm of its coefficients.
  We then use the following lemma 
  which turns any bound on the coefficients of a polynomial in $\tau$
  into a bound on the expectation
  of the absolute value of this polynomial, where the
  expectation is over $\tau$ sampled from a large set (in this case $\cDq$).
  \begin{lemma}[Lemma 5.5 of \cite{DBLP:conf/stoc/FeldmanPV15}]
    \label{lm:santosh}
    Let $p(x)$ be a polynomial of degree $\ell$ in $\{-1, 1\}^n$,
    and let $\cS \subseteq \{-1, 1\}^n$ be a set of assignments.
    The following holds for $d = 2^n / \abs{\cS}$
    \begin{align*}
      \Exu{\tau \sim \cS}{\abs{p(\tau)}} \le O_\ell\rbr{\rbr{\ln d}^{\ell / 2} \norm{p}_2}.
    \end{align*}
  \end{lemma}

  We now proceed with a proof.
  For any $A \subseteq[n]$ of size $\ell$ we define
  \begin{align*}
      u_{S, A} = 
      \frac{1}{\abs{\cI}}
      \sum_{I: I[S] = A} \sum_{x\in \{-1, 1\}^k}
      \hat{Q}_I(S) \chi(x[S]) h(I, x),
  \end{align*}
  Grouping Equation~\eqref{eq:M_S_sum} based on $I[S]$ we obtain
  $M_S(\tau) = \sum_{A} u_{S, A} \chi(\tau[A]) = \sum_{A} u_{S, A} \prod_{i \in A} \tau[i]$,
  which verifies that $M_S(\tau)$ is a polynomial in $\tau$.

  We proceed to bound $\norm{M_S}_2 = \sqrt{\sum_{A} u_{S, A}^2}$.
  For any fixed $A$, using the Cauchy–Schwarz inequality\redit{,} we obtain
  \begin{align}
    \abs{\cI}^2 u_{S, A}^2
    &=
    \rbr{
      \sum_{I: I[S] = A}
      \sum_{x}
      \hat{Q}_{I}(S) \chi(x[S])
      h(I, x)
    }^2
    \notag
    \\&\le
    \rbr{\sum_{I: I[S] = A}
    \sum_{x}
    \hat{Q}^2_{I}(S) \chi^2(x[S])}
    \rbr{
      \sum_{I: I[S]  = A}
      \sum_{x} 
      h^2(I, x)
    }
    .
    \label{eq:u_S_A}
  \end{align}

  We start by bounding the first term.
  By Parseval's identity,
  \begin{align*}
    \hat{Q}_I^2(S)
    &\le 
    \sum_{S'}
    \hat{Q}_I^2(S')
    &\EqComment{Non-negativity of $Q_{I}^2(S')$}
    \\&=
    \Exu{x \sim \{-1, 1\}^k}
    {Q_I^2(x)}
    &\EqComment{Parseval's identity}
    \\&\le 
    \Exu{x \sim \{-1, 1\}^k}{1}
    &\EqComment{Since $Q_I(x) \le 1$}
    \\&=1. 
  \end{align*}
  Since $\chi^2(x[S]) = 1$, this implies
  \begin{align*}
    \sum_{I: I[S] = A}
    \sum_{x}
    \hat{Q}^2_{I}(S) \chi^2(x[S])
    &\le 
    \sum_{I: I[S] = A}
    \sum_{x}
    1
    =
    \abs{\cbr{I: I[S] = A}} 2^k
    ,
  \end{align*}
  We note however that since the constraint $\cbr{I[S] = A}$
  groups all values of $I$ in terms of $I[S]$ and
  the number of possible values for $I$ and
  $I[S]$ is $\abs{\cI_k}$ and $\abs{\cI_\ell}$ respectively,
  \begin{math}
    \abs{\cbr{I: I[S] = A}} = \frac{\abs{\cI_k}}{\abs{\cI_\ell}}.
  \end{math}
  Plugging this back in
  \eqref{eq:u_S_A}
  we obtain
  \begin{align*}
    \abs{\cI_k}^2
    \norm{M_S}_2^2
    =
    \abs{\cI_k}^2
    \sum_{A} u_{S, A}^2
    &\le
    2^k\frac{\abs{\cI_k}}{\abs{\cI_\ell}}
    \rbr{
      \sum_{A} \sum_{I: I[S] = A} \sum_{x} h^2(I, x)
    }
    \\&=
    2^k\frac{\abs{\cI_k}}{\abs{\cI_\ell}}
    \rbr{
      \sum_{I, x} \sum_{A=I[S]} h^2(I, x)
    }
    &\EqComment{Rearranging}
    \\&=
    2^k\frac{\abs{\cI_k}}{\abs{\cI_\ell}}
    \rbr{
      \sum_{I, x} h^2(I, x)
    }
    &\EqComment{Uniqueness of $I[S]$}
    \\&\le
    2^{2k}\frac{\abs{\cI_k}}{\abs{\cI_\ell}}
    \abs{\cI_k}
  \end{align*}
  where the last inequality uses the fact that 
  \begin{align*}
    \frac{\sum_{I, x} h^2(I, x)}{\abs{X_k}}
    = \Ex{h^2(I, x)} = 
    \norm{h}_2^2
    \le 1
    \implies 
    \sum_{I, x} h^2(I, x) \le \abs{X_k}
     = 2^k \abs{\cI_k} 
     .
  \end{align*}
  It follows that
  \begin{math}
    \norm{M_S}_2 \le \frac{2^k}{\sqrt{\abs{\cI_\ell}}} .
  \end{math}
  This bounds the norm of the coefficients of $M_S$
  which is a polynomial in $\tau$. 
  By assumption in the theorem, the set $\cD'$ has size $\ge \abs{\cD}/q$.
  Therefore, Lemma~\ref{lm:santosh} implies that
  \begin{align*}
    \Exu{\cQ_{\tau} \sim \cD'}{\abs{M_S(\tau)}}
    \le O_k\rbr{
      \frac{\ln(q)^{\ell/2}}{\sqrt{\abs{\cI_\ell}}}
    }
    = O_k\rbr{\frac{\ln(q)^{\ell/2}}{n^{\ell / 2}}}
  \end{align*}
  Summing over $S$
  and plugging back in Equation~\eqref{eq:sdn_bound_decompose},
  \begin{align*}
  \Exu{\cQ_\tau \sim \cD'}{
  \abs{\Exu{\cQ_{\tau}}{h} - \Exu{\cU}{h}}
  }
  \le O_k\rbr{\frac{\ln(q)^{r/2}}{n^{r / 2}}}
  ,
  \end{align*}
  which finishes the proof.
\end{proof}

\section{Worst-case analysis of SAT instances with multiple solutions - Proof of Theorem~\ref{thm:worst-case-main}}
\label{sec:worst-case}

In this section, we formalize and analyze the computational complexity of satisfiability problems (SAT) when a large number of solutions is guaranteed. Specifically, we consider SAT instances where we know in advance that the instance either has no solutions or has at least a certain number of solutions.

\begin{definition}
Let $\SAT(n, s)$ denote the class of SAT instances
with \(n\) variables, where it is known a priori that the instances have either zero solutions or at least \(s\) solutions. 
The class of all SAT instances over $n$ variables
can be denoted as \(\SAT(n, 1)\).
\end{definition}

\begin{theorem}
\label{thm:wc}
Assuming the Exponential Time Hypothesis (ETH), there is no polynomial-time algorithm for solving \(\SAT(n, 2^{n - \log(n)^{1+c}})\) for any \(c > 0\).
\end{theorem}

This statement underscores the profound computational challenge presented by SAT instances with exponentially many solutions, delineating a clear boundary in the solvability landscape.
It is important to note that we
say an algorithm $\mathcal{A}$ can \emph{solve} $\SAT(n, s)$ if, for any formula $f$ from this class, $\mathcal{A}$ can determine whether $f$ is satisfiable or unsatisfiable.

\begin{proof}
We construct a polynomial-time reduction from an arbitrary instance of \(\SAT(\log(n)^{1+c}, 1)\), known to be hard to solve in subexponential time under the Exponential Time Hypothesis (ETH), to \(\SAT(n, 2^{n - \log(n)^{1+c}})\). By augmenting the \(\SAT(\log(n)^{1+c}, 1)\) instance with \(n - \log(n)^{1+c}\) additional variables, we enable each solution of the original problem to pair with \(2^{n - \log(n)^{1+c}}\) assignments of the new variables, thus forming a \(\SAT(n, 2^{n - \log(n)^{1+c}})\) instance. Denote the size of the subinstance as \(m = \log(n)^{1+c}\), leading to \(n = 2^{m^{1/(1+c)}}\). The existence of a polynomial-time algorithm for \(\SAT(n, 2^{n - \log(n)^{1+c}})\), or \(O(n^k)\), implies an algorithm for the subproblem within \(O(2^{k \cdot m^{1/(1+c)}})\), which constitutes a subexponential time complexity. This contradicts the ETH by suggesting a subexponential time solution to a problem that is presumed to require exponential time, thereby affirming the infeasibility of such a polynomial-time algorithm for the given SAT formulation.
\end{proof}

\begin{claim}
The introduction of a large solution space in \(\SAT(n, 2^{n - \log(n)^{1+c}})\) does not simplify the problem, preserving its computational complexity akin to that of general SAT instances.
\end{claim}

Furthermore, we demonstrate the tightness of our bound by considering the scenario where \(c \geq 0\), specifically for \(c = 0\), which implies \(2^{n-\log(n)}\) solutions for the SAT instance. This scenario offers an intriguing contrast: a polynomial-time algorithm becomes feasible for \(\SAT(n, 2^{n-\log(n)})\).

To elucidate, consider a random assignment of variables in a \(\SAT(n, 2^{n-\log(n)})\) problem. The probability \(P\) that a random assignment satisfies the instance is given by the ratio of the number of satisfying solutions to the total number of possible assignments:

\[
P = \frac{2^{n-\log(n)}}{2^n} = \frac{1}{2^{\log(n)}} = \frac{1}{n}
\]

Thus, the probability of a random assignment not satisfying the instance is \(1 - \frac{1}{n}\). To find a satisfying assignment, we can repeatedly select random assignments. After \(k\) independent trials, the probability that all trials fail to find a satisfying assignment is 
\(
(1 - \frac{1}{n})^k
\).
To ensure a high probability of success, we aim for this probability to be less than a certain threshold, say \(\epsilon\), where \(\epsilon > 0\). To determine the required number of trials, we select \(k\) satisfying 
\(
(1 - \frac{1}{n})^k < \epsilon\redit{,}
\)
 For a sufficiently large \(n\), we can approximate \(1 - \frac{1}{n}\) as \(e^{-\frac{1}{n}}\), giving us:
\(
e^{-\frac{k}{n}} < \epsilon
\)
Solving for \(k\), we get:
\(
k > n \log(\frac{1}{\epsilon})
\)

Since \(k\) scales linearly with \(n\), a polynomial number of executions (specifically, \(O(n)\)) suffices to provide a YES/NO answer to \(\SAT(n, 2^{n-\log(n)})\) with high probability, thereby illustrating that the bound of \(2^{n-\log(n)^{1+c}}\) for \(c > 0\) is tight. Reducing \(c\) to \(0\) shifts the problem into a domain where polynomial-time solutions emerge, underscoring the critical nature of the \(c > 0\) condition in maintaining the problem's hardness under our framework.

\color{black}

Furthermore, this technique can be utilized to demonstrate the hardness associated with specific restrictions. For example, in \cite{certsol} Hsieh, Mohanty, and Xu propose the open problem of distinguishing between 3-SAT instances that are either unsatisfiable or have at least \(T\) solutions, under the stipulation that the input is guaranteed to be (7/8 + \(\epsilon\))-satisfiable (refer to Question 1.22 in \cite{certsol}).

We can directly respond to and resolve this query by extending Theorem~\ref{thm:wc}:

\begin{theorem}
Assuming the Exponential Time Hypothesis (ETH), there is no polynomial-time algorithm for solving \(\SAT(n, 2^{n - \log(n)^{1+c}})\) for any \(c > 0\), even when it is guaranteed that the input is (7/8 + \(\epsilon\))-satisfiable.
\end{theorem}

\begin{proof}
As noted by Hsieh et al. \cite{certsol}, distinguishing between a (7/8 + \(\epsilon\))-satisfiable 3-SAT instance and a fully satisfiable 3-SAT instance is NP-hard. Following the framework of Theorem~\ref{thm:wc}, we consider a reduction from an instance of this problem to a \(\SAT(n, 2^{n - \log(n)^{1+c}})\) instance. By adding extra variables, we expand the solution space while preserving the (7/8 + \(\epsilon\)) satisfiability condition. The addition of variables exponentially increases the number of solutions only if the original instance is already satisfiable, without altering the unsatisfiability of cases where no solutions exist. Consequently, the complexity of distinguishing between these markedly different states—completely unsatisfiable versus possessing a multitude of solutions—remains NP-hard under the ETH.
\end{proof}

It is also important to note that by maintaining the (7/8 + \(\epsilon\)) satisfiability condition during the reduction, all related justifications concerning this transformation are valid. This includes the argument around the tightness of this bound and the existence of a randomized algorithm for instances with at least $2^{n-\log(n)}$ solutions. 

\end{document}